%% file: tmp/sample-sigplan.tex
\documentclass[sigconf,screen,nonacm]{acmart}

\setcopyright{none}

\usepackage{enumitem}
\usepackage{multirow}
\usepackage{booktabs}
\usepackage{caption}
\usepackage{array}
\usepackage{makecell}
\usepackage{color}
\usepackage{xcolor}

\usepackage{subcaption}
\usepackage{algorithm}
\usepackage{algpseudocode}
\usepackage{threeparttable}
\usepackage{pifont}
\usepackage{pdfpages}
\usepackage{pdfcomment}
\usepackage{tikz}
\usepackage{amsmath}
\usepackage{pifont}
\usepackage{float}
\usepackage{bm}
\usepackage{listings} 
\usepackage{xcolor} 
\usepackage{stfloats}
\usepackage{adjustbox}
\usepackage{graphicx,booktabs}
\usepackage{booktabs,threeparttable}
\usepackage[table]{xcolor}
\usepackage{tabularx}
\usepackage{url}
\usepackage{hyperref}
\usepackage{placeins}
\newcolumntype{N}{>{\centering\arraybackslash}m{3.1em}}
\usepackage{tocloft}
\usepackage{etoc}
\usepackage{mathtools}
\usepackage{longtable}
\usepackage{graphicx}
\usepackage{hyperref}

\AtBeginDocument{%
  }

\input{macro}
\renewcommand\footnotetextcopyrightpermission[1]{} 

\begin{document}

\title{\method{}: Multimodal Corroboration of Latent Asset Signals for Financial Trading}


\author{Yanzheng Jin}
\affiliation{%
 \institution{National University of Singapore}
 \country{Singapore}
}

\author{Pengyang Shao}
\affiliation{%
 \institution{National University of Singapore}
 \country{Singapore}
}

\author{Xiaohao Liu}
\affiliation{%
 \institution{National University of Singapore}
 \country{Singapore}
}

\author{Xi Ai}
\affiliation{%
 \institution{National University of Singapore}
 \country{Singapore}
}

\author{Fei Shen}
\affiliation{%
 \institution{National University of Singapore}
 \country{Singapore}
}

\author{Kenji Kawaguchi}
\affiliation{%
 \institution{National University of Singapore}
 \country{Singapore}
}

\renewcommand{\shortauthors}{Yanzheng Jin et al.}

\begin{abstract}

Financial trading relies on extracting reliable signals from heterogeneous market modalities such as price series, breaking news, and investor sentiment. Existing multimodal methods primarily combine heterogeneous modalities to exploit complementarity, treating each modality as equally valuable while overlooking whether different modalities provide mutually supportive evidence for the same trading signal.
However, this task-conditioned and non-canceling support, termed \textit{multimodal corroboration}, is particularly valuable, especially for financial trading. 
Because individual financial views are noisy and weakly informative, support that persists across heterogeneous views may provide a more stable task-relevant signal than evidence appearing in only one view.
To exploit this property, we propose \method{} (multimodal \underline{\textbf{Co}}rroboration of \underline{\textbf{L}}atent \underline{\textbf{A}}sset \underline{\textbf{S}}ignals), a framework that operationalizes multimodal corroboration as a trainable task-conditioned representation for trading prediction.
The modality representations are organized into a per-instance matrix, where a softmax-based spectral objective strengthens its dominant shared component. Signed modality contributions then determine whether this component provides non-canceling support and construct the resulting corroborated signal.
A coupled robustness-aware consistency objective further preserves the resulting corroborated signal when a modality is corrupted or missing. 
Extensive experiments on stock and cryptocurrency datasets demonstrate the effectiveness of our proposed \method{}, yielding consistent improvements in both annualized return and Sharpe ratio over existing methods.

\end{abstract}





\renewcommand{\addcontentsline}[3]{}

\maketitle

\input{tex/1introduction}
\input{tex/2related_work}
\input{tex/3method}
\input{tables/baselines}
\input{tex/4experiment}

\input{tex/5conclusion}

\bibliographystyle{ACM-Reference-Format}
\balance
\bibliography{references}

\appendix
\input{tex/6appendix}

\end{document}

%% file: macro.tex
\newcommand{\best}[1]{\textbf{#1}}
\newcommand{\second}[1]{\underline{#1}}

\newcommand{\highlight}[1]{\textbf{#1}}

\newcommand{\ignore}[1]{}

\newcommand{\method}[1]{CoLAS}

\newcommand{\ARRhead}{\shortstack{\textbf{ARR\%$\uparrow$}}}
\newcommand{\SRhead}{\shortstack{\textbf{SR$\uparrow$}}}

\newcommand{\cmark}{\ding{51}}
\newcommand{\xmark}{\ding{55}}

%% file: tex/1introduction.tex
\section{Introduction}
\label{sec:introduction}
Data-driven financial trading has attracted increasing research interest, due to its ability to discover reliable trading signals from various market data~\cite{xu2018stock}. Typically, these methods take a window of historical market data as input and output a trading decision, such as a bullish or bearish signal~\cite{zhang2024multimodal,xiao2024tradingagents}. 
Among such data, heterogeneous modalities, \textit{e.g.}, price series, breaking news, and investor sentiment, characterize the market from diverse perspectives, providing reliable signals.

Existing work largely centers on \textit{how to exploit the available market modalities to extract trading signals}. Early approaches combine modalities at a coarse, input level: a representative line builds LLM-based trading agents that place heterogeneous inputs, \textit{e.g.}, news, prices, and charts, together in the language model's context and rely on its reasoning to reach a decision~\cite{zhang2024multimodal,xiao2024tradingagents,li2026time}. 
They distribute different information sources to specialized agents and aggregate their opinions into a final action.
More recent methods move beyond such prompt-level combination in two directions. The first develops dedicated multimodal architectures that combine modality representations through specialized modules to exploit complementary information, the information one modality carries but another lacks~\cite{koa2025reasoning}. The second turns to large-scale pretraining, building time-series foundation models that capture rich temporal patterns and attain strong performance from price and technical signals alone~\cite{shi2026kronos}. 
Despite their differences, these approaches mainly emphasize aggregating information to exploit complementary evidence, rather than explicitly modeling whether modality contributions reinforce or cancel a shared predictive direction. 

\begin{figure}[t]
    \centering
    \includegraphics[width=\columnwidth]{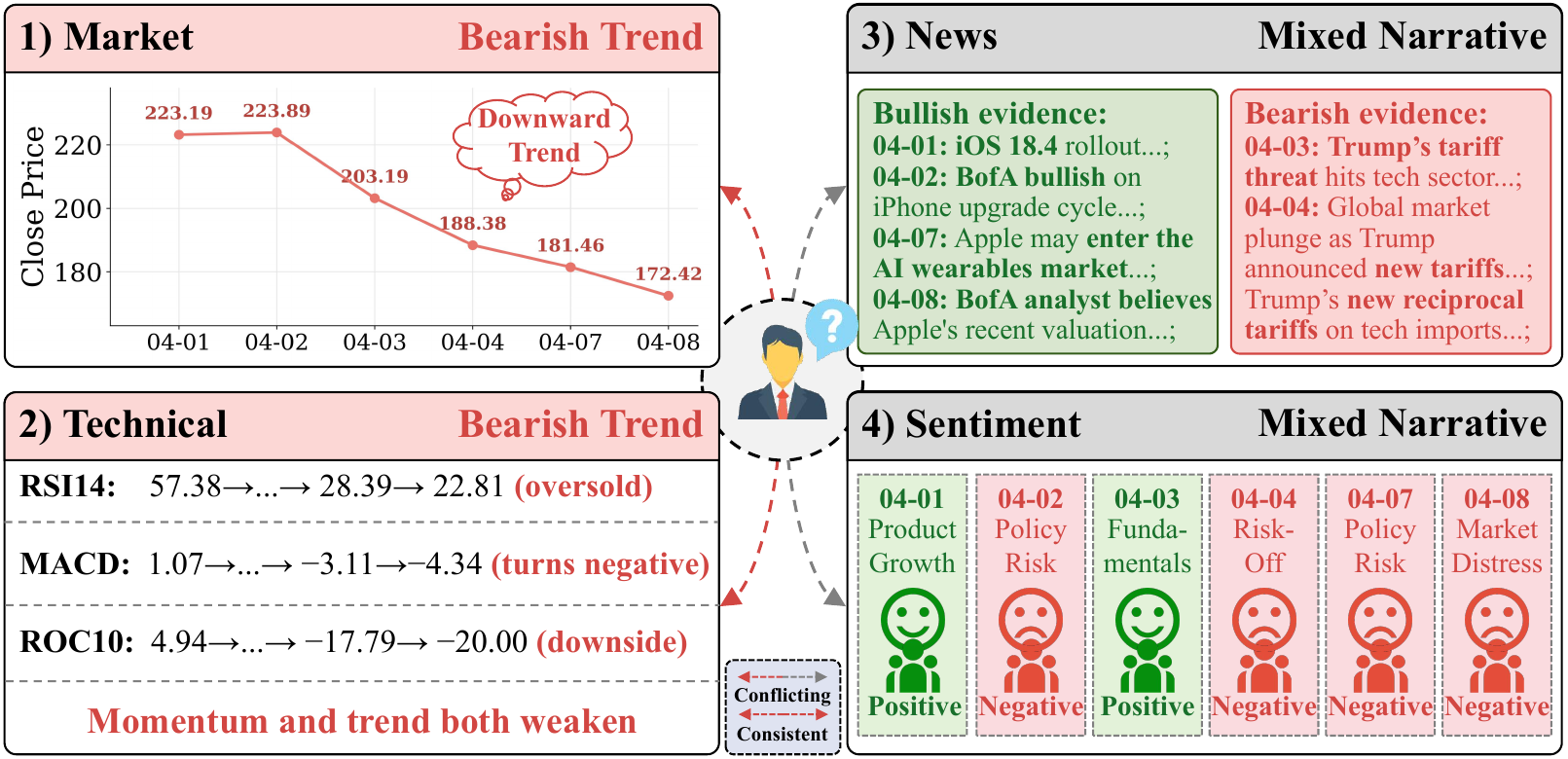}
    \vspace{-6mm}
    \caption{
    An illustration showing limitations of existing single- and multimodal financial trading methods.}
    \vspace{-6mm}
    \label{fig:intro}
\end{figure}

This distinction overlooks that multimodal information is not uniformly valuable for financial trading: \textit{what a modality contributes depends not on its mere presence, but on how it relates to the others}.
As illustrated in Figure~\ref{fig:intro}, market prices and technical indicators can jointly point to a bearish trend, whereas news and sentiment may instead provide bullish evidence. This divergence motivates a finer decomposition of multimodal information into three components.
The first, corroboration, refers to task-conditioned, non-canceling support provided by heterogeneous financial modalities for a common predictive direction.
The second, complementarity, is the information one modality contributes beyond the others. The third, conflict, arises when modalities provide opposing support for the predictive direction. Among the three, corroboration is an important and underexplored component for financial trading, for two reasons. (1) 
Corroboration across heterogeneous financial views is less susceptible to modality-specific noise than a signal appearing in only one view~\cite{mathur2022docfin,lee2025hierarchical}. This does not require statistical independence among the views. Market observations, transformed quantitative indicators, event-level news, and investor sentiment are generated or represented through different processes and exhibit distinct temporal characteristics and noise patterns. A direction retained across these heterogeneous views is therefore less likely to be driven solely by the idiosyncratic noise of one modality.
Furthermore, (2) such corroboration is especially valuable in finance~\cite{gu2020empirical}, where exploitable signals are weak and each individual view is inherently noisy~\cite{fama1970efficient,ctictan2015efficient}. Price series are dominated by short-term fluctuations, technical indicators may lag or amplify transient movements, news arrives irregularly, and investor sentiment is volatile and easily confounded~\cite{li2020multimodal,zheludev2014can,audrino2020impact}. Consequently, single-view predictions can be fragile, whereas a common task-relevant direction across views provides a more stable basis for financial trading.

This motivates us to explicitly model multimodal corroboration rather than entangling shared structure, whether supportive evidence or cross-modal conflict, into a single fused vector, which leaves potential corroboration implicit and difficult to recover~\cite{zhou2025triple,zhao2025resolving,hazarika2020misa}.
We further require reliable corroboration even when an individual modality becomes noisy or unavailable. These two principles, explicit corroboration modeling and reliability under modality degradation, underlie our proposed method.

To this end, we introduce \textbf{\method{}} (multimodal \underline{\textbf{Co}}rroboration of \underline{\textbf{L}}atent \underline{\textbf{A}}sset \underline{\textbf{S}}ignals) for financial trading, to capture the task-conditioned corroborated signal from four important modalities, including the market, technical, news, and sentiment streams. 
\method{} encodes each stream with a modality-specific encoder and projects it into a shared latent space, where the modalities of the same instance become directly comparable. 
A corroboration mining objective then operationalizes multimodal corroboration through two complementary components. 
First, singular value maximization concentrates the modality representations on a candidate shared spectral component, while signed modality contributions determine whether this component provides effective, non-canceling support. 
Second, the resulting corroborated signal is further regularized across instances to preserve sample discrimination, so the corroboration of different instances does not collapse together. 
In addition, a coupled robustness-aware consistency objective encourages the corroborated signal to remain stable when an individual modality is corrupted or missing. 
The two objectives are optimized jointly with the movement-prediction loss. 
The main contributions of \method{} are summarized as follows:
\begin{itemize}[left=0em]
\item We identify multimodal corroboration as task-conditioned, non-canceling support among heterogeneous financial modalities, as an important and underexplored factor for financial trading. We further explain why it provides a strong and reliable signal under the variety and noise of financial data.
\item We propose \method{}, a novel framework that operationalizes multimodal corroboration as a trainable representation mechanism. 
We strengthen the dominant shared spectral component for corroboration mining and regularize this signal to remain discriminative across instances. 
A coupled consistency objective further keeps it reliable when a modality is corrupted or missing.
\item We conduct extensive experiments on stock and cryptocurrency datasets, demonstrating the effectiveness of \method{} in delivering both higher returns and better risk-adjusted performance. Across the six primary datasets, \method{} improves annualized return over 25\% on average relative to the strongest baseline in each dataset. It also achieves the highest Sharpe ratio on all six datasets.
\end{itemize}

%% file: tex/2related_work.tex
\section{Related Work}
\label{sec:related_work}
\subsection{AI for Financial Trading}
Graph-based stock predictors have modeled temporal autocorrelation~\cite{wang2025gcn} and inter-stock relationships~\cite{lu2025leveraging}.
The advent of Large Language Models (LLMs) has significantly advanced financial analysis by introducing semantic reasoning capabilities~\cite{wu2023bloomberggpt}, yet challenges in multimodal integration persist~\cite{chen2024deep}. 
Early studies~\cite{chen2023chatgpt} leveraged general-purpose LLMs (e.g., GPT-4) for zero-shot sentiment analysis and market prediction, employing prompt engineering techniques such as Chain-of-Thought (CoT)~\cite{wei2022chain} to extract trading signals from textual data~\cite{lopez2023can,wei2022chain,hansen2024can}. 
However, without domain-specific fine-tuning, these general LLMs often hallucinate on quantitative tasks and struggle to align textual sentiment with precise market movements. To bridge this domain gap, researchers developed Financial LLMs such as FinGPT~\cite{liu2023fingpt} and BloombergGPT~\cite{wu2023bloomberggpt}, which are fine-tuned on extensive financial corpora to enhance domain alignment. Building upon these foundations, the field has recently advanced towards agentic trading, where LLMs function as autonomous agents equipped with memory, reflection, and tool-use capabilities~\cite{yu2025finmem,zhang2024multimodal,xiao2024tradingagents,li2026time}. Despite enabling dynamic interaction with market environments, most existing agents remain predominantly text-centric. They tend to process numerical and textual modalities in isolation, a limitation that the proposed framework \method{} explicitly addresses.

\subsection{Multimodal Alignment}
Multimodal alignment embeds heterogeneous modalities into a shared space where representations of the same instance are close. Prior work is bimodal, CLIP~\cite{radford2021learning} aligns vision and language through pairwise contrastive learning, pulling matched pairs together while pushing mismatched ones apart. This paradigm has been extended to audio-text, point-text, and video-text pairs~\cite{elizalde2023clap,zhang2022pointclip,luo2022clip4clip}, and then to unified foundation models that fold in additional modalities such as audio~\cite{ruan2023accommodating,guzhov2022audioclip} and subtitles~\cite{luo2020univl,li2021value}. Besides these innovations, anchor-based methods designate one modality as a reference to align several modalities at once. Particularly, ImageBind~\cite{girdhar2023imagebind} binds all modalities to vision, while other models anchor on language or point clouds~\cite{zhu2024languagebind,guo2023point}. Framing every modality through one fixed reference, however, constrains how richly the modalities can interact. 
In contrast, another category aligns all modalities simultaneously rather than pairwise. 
A recent method~\cite{cicchetti2025gramian} proposes aligning multiple modalities simultaneously from a geometrical perspective. 
However, these methods are primarily developed for settings in which the main objective is semantic alignment across modalities. Financial trading additionally requires handling noisy and potentially opposing evidence that varies across market states, and an additional question is how to distinguish supportive from opposing modality contributions for a given prediction. 
Building upon this intuition, we propose \method{}, which combines multimodal corroboration mining with a robust prediction layer to construct a task-conditioned representation that combines dominant cross-view structure with signed modality support for financial trading.

%% file: tex/3method.tex
\input{tables/multimodal_data}
\section{Preliminary}
\label{sec:preliminary}
\noindent
\highlight{Problem formulation.} We study the task of predicting an asset's next-day directional signal from a historical window of $T$ trading days. 
Each instance contains four temporally aligned financial modalities, namely market, technical, news, and sentiment, as summarized in Table~\ref{tab:data_modalities}. These modalities may share upstream information but differ in their information content, construction processes, and noise characteristics (see details in Appendix~\ref{app:datasets}).
For each trading day
$t$ and modality $m\in\mathcal{M}$, the model input is defined as $\mathbf{X}_t^m
=
\left\{
\mathbf{x}_{t-T}^m,\ldots,\mathbf{x}_{t-1}^m
\right\},$ which contains only observations from the $T$ trading days preceding
day $t$ and excludes all observations from day $t$. Given the multimodal history
$\{\mathbf{X}_t^m\}_{m\in\mathcal{M}}$, the model predicts the
subsequent close-to-close price direction: $y_t
=
\mathbb{I}
\left[
p_{t+1}^c\geq p_t^c
\right],$
where $p_t^c$ denotes the adjusted closing price on trading day $t$.
Here, $y_t=1$ denotes a bullish signal from day $t$
to day $t+1$, whereas $y_t=0$ denotes a bearish signal.

\noindent
\highlight{Multimodal alignment.} We consider an instance with $k=|\mathcal{M}|$ modalities, whose normalized representations are stacked column-wise into a matrix $\mathbf{V}\in\mathbb{R}^{d_c\times k}$. Multimodal alignment seeks representations that are mutually consistent across modalities. The pairwise contrastive learning paradigm extends to $k>2$ by decomposing the objective over modality pairs. 
For example, the training on ${m_1,m_2,m_3}$ can be decomposed to the training sequences of $\mathcal{L}_{m_1,m_2}$, $\mathcal{L}_{m_1,m_3}$, and $\mathcal{L}_{m_2,m_3}$~\cite{girdhar2023imagebind,chen2023vast}.
Aligning all modalities jointly, rather than pair by pair, can instead be expressed through the Gram matrix $\mathbf{G}=\mathbf{V}^\top\mathbf{V}\in\mathbb{R}^{k\times k}$~\cite{cicchetti2025gramian,liu2026principled}.
The gains from simultaneous alignment motivate us to examine the spectral structure of $\mathbf{V}$ more closely.

\noindent
\highlight{Singular Value Decomposition (SVD).} For a matrix $\mathbf{X}$, the singular value decomposition is $\mathbf{X}=\mathbf{P}\mathbf{\Sigma}\mathbf{R}^\top$, where $\mathbf{P}$ and $\mathbf{R}$ are orthogonal and $\mathbf{\Sigma}$ carries the non-negative singular values $\lambda_1\ge\lambda_2\ge\cdots\ge 0$ on its diagonal. These singular values are the square roots of the eigenvalues of $\mathbf{X}^\top\mathbf{X}$, so the largest one satisfies $\lambda_1^2=\sigma_1$, the dominant eigenvalue, and its left singular vector defines the dominant direction. Applied to the modality matrix V, the leading singular direction identifies the axis along which the learned modality representations are most geometrically concentrated. This spectral concentration is used as a proxy for dominant shared cross-view structure, but is not by itself sufficient to establish effective corroboration. The principal notation is summarized in Appendix~\ref{app:notations}.

\begin{figure*}[t]
    \centering
    \includegraphics[width=0.95\textwidth]{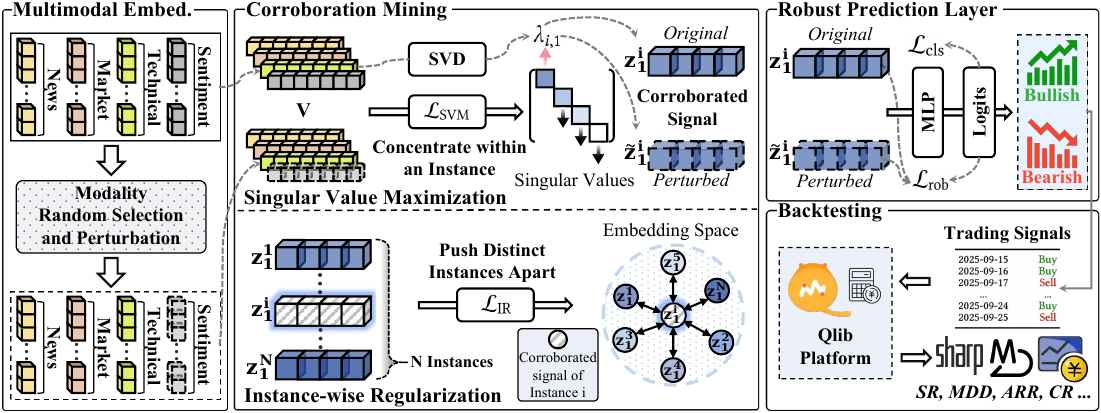}
    \vspace{-3mm}
    \caption{The overall framework of \method{}. Four modality streams are encoded into a shared alignment space and stacked into a per-instance matrix $\mathbf{V}$. \textit{Corroboration Mining} 
    maximizes the leading singular value of $\mathbf{V}$, concentrating the modality representations on a dominant spectral component, while signed modality contributions determine the resulting corroborated signal $\mathbf{z}$
    , and regularizes these directions across instances. 
    \method{} adopts a \textit{Robust Prediction Layer} via perturbation to enhance the prediction robustness. Finally, the resulting signals are backtested.}
    \vspace{-3mm}
    \label{fig:overall_structure_section}
\end{figure*}

\section{METHODOLOGY}
\label{sec:method}
Our premise is simple: a trading signal is reliable when heterogeneous financial modalities corroborate it, and fragile when it rests on any one modality. \method{} builds upon this premise to summarize task-relevant corroboration among heterogeneous financial inputs, as illustrated in Figure~\ref{fig:overall_structure_section}. It operationalizes corroboration through the spectral structure and signed modality support of a per-instance representation matrix, and comprises four modules: \textbf{multimodal embedding} (\textit{c.f.}, Sec.~\ref{sec:encoder}), \textbf{corroboration mining} (\textit{c.f.}, Sec.~\ref{sec:svm}), \textbf{robust prediction layer} (\textit{c.f.}, Sec.~\ref{sec:robust}), and \textbf{backtesting} (\textit{c.f.}, Appendix~\ref{app:backtest}). 

\vspace{2mm}
\noindent\textbf{Overview}. \emph{Multimodal embedding} maps the heterogeneous inputs into a common alignment space and stacks their normalized representations into a per-instance matrix. \emph{Corroboration mining} then identifies a dominant spectral component and aggregates the signed modality projections on this component. The resulting signal is strong only when the component is both spectrally concentrated and non-canceling across modalities. A cross-instance regularizer preserves sample discrimination, while the \emph{robust prediction layer} maintains the resulting signal under single-modality degradation. Finally, the predictions are converted into trading decisions and evaluated through \emph{backtesting}~\cite{yang2020qlib}.

\subsection{Modality-specific Financial Encoders}
\label{sec:encoder}
Financial modalities differ in their statistical structure, so we encode them with modality-aware networks rather than a single shared encoder. We describe the encoding for one sampled instance and omit its index for brevity. Each modality $m \in \mathcal{M}$ is processed by a
modality-specific encoder $f_m$, which maps the corresponding input window to a modality-specific embedding $\mathbf{h}^{m}=f_m\left(\mathbf{X}^{m};\boldsymbol{\theta}_m\right),$
where $\boldsymbol{\theta}_m$ are the encoder parameters and $\mathcal{M}$ is the modality set introduced in Sec.~\ref{sec:preliminary}. 
We instantiate the market and technical encoders as LSTMs that summarize the temporal dynamics of the two streams. For the news and sentiment streams, whose inputs are precomputed daily embeddings (see Appendix~\ref{app:datasets}), the corresponding encoders aggregate the embeddings over the window and apply lightweight projection heads.

We then bring all modalities into a common geometry in which shared structure and signed modality support can be characterized. 
To this end, each embedding $\mathbf{h}^{m}$ is mapped into a shared alignment space $\mathbb{R}^{d_c}$ by a modality-specific alignment map $\psi_m$ and $\ell_2$-normalized onto the unit sphere: $\mathbf{v}^{m}=\psi_m(\mathbf{h}^{m}),
\bar{\mathbf{v}}^{m}=\frac{\mathbf{v}^{m}}{\lVert \mathbf{v}^{m} \rVert_2}.$
Stacking the normalized vectors column-wise gives:
\begin{equation}
\mathbf{V}=\big[\bar{\mathbf{v}}^{m_1},\bar{\mathbf{v}}^{m_2},\ldots,\bar{\mathbf{v}}^{m_k}\big]\in\mathbb{R}^{d_c\times k},
\qquad k=|\mathcal{M}|.
\end{equation}
Every column of $\mathbf{V}$ lies on the unit sphere. The modality representations share a common scale and become directly comparable, providing the basis for extracting their shared spectral structure and signed modality support.

\subsection{Multimodal Corroboration Mining}
\label{sec:svm}
\emph{Multimodal corroboration} is represented by two factors: spectral concentration and net signed support. Spectral concentration measures whether the modality representations admit a dominant shared component, while net signed support measures whether their projections on this component reinforce or cancel one another. The two factors jointly determine the strength of the resulting corroborated signal. Their formal relationship is summarized in this section and proved in Appendix~\ref{app:theoretical_analysis}.

\vspace{2mm}
\noindent\highlight{Corroboration enhancement via singular value maximization.}
Recall that the normalized modality vectors of a sample are stacked column-wise into $\mathbf{V}\in\mathbb R^{d_c\times k}$. We now reinstate the instance index $i$ for the batch and obtain the singular spectrum of $\mathbf{V}_i$ by SVD~\cite{golub2013matrix,wall2003singular,abdi2007singular}, \textit{i.e.},$\mathbf{V}_i=\mathbf{P}_i\mathbf{\boldsymbol\Lambda}_i\mathbf{R}_i^{\top},\boldsymbol\Lambda_i=\mathrm{diag}\left(\lambda_{i,1},\lambda_{i,2},\ldots,\lambda_{i,k}\right),$ where the columns of $\mathbf{P}_i$ are left singular vectors, the columns of $\mathbf{R}_i$ are right singular vectors, and $\lambda_{i,1}$ is the largest singular value.

To concentrate the multimodal representations on a dominant shared
component, we treat the singular values as logits and apply a
softmax-based objective:
\begin{equation}
\mathcal L_{\mathrm{SVM}}
=
-\frac{1}{B}
\sum_{i=1}^{B}
\log
\frac{
\exp(\lambda_{i,1}/\gamma)
}{
\sum_{j=1}^{k}
\exp(\lambda_{i,j}/\gamma)
},
\label{eq:svm_loss}
\end{equation}
where $B$ is the batch size and $\gamma$ is a temperature parameter.
Because every column of $\mathbf{V}_i$ is unit-normalized, its total spectral energy is fixed: $\sum_{j=1}^{k}\lambda_{i,j}^2=k$. Under this fixed-energy constraint, transferring spectral energy from a non-leading component to the leading component decreases $\mathcal L_{\mathrm{SVM}}$ together with the residual energy outside the leading rank-1 component, as established in Proposition~\ref{prop:spectral_concentration}. We define $\rho_i=\lambda_{i,1}^2/k$ as the spectral concentration factor. A larger $\rho_i$ indicates that a greater proportion of the modality representations is captured by the dominant component. We next use the signed modality contributions along this dominant component to construct the corroborated signal.

\noindent\highlight{Signed corroborated signal.}
Let $\mathbf{V}_i^{(1)}=
\mathbf{p}_{i,1}
\lambda_{i,1}
\mathbf{r}_{i,1}^{\top}$
denote the leading rank-1 component of $\mathbf{V}_i$. The signed contribution of modality $m$ along the leading direction is $a_{i,m} = \mathbf{p}_{i,1}^{\top} \bar{\mathbf{v}}_i^m$, where a positive or negative value indicates whether the modality supports or opposes the selected orientation of the leading component. Let $\bar{\mathbf{v}}_i=k^{-1}\sum_{m=1}^{k}\bar{\mathbf{v}}_i^m$ (see Eq.~\ref{eq:app_defs}) denote the modality mean. The corroborated signal is its projection onto the leading direction: $\mathbf{z}_i=\mathbf{{proj}_{{p}_{i,1}}(\bar{\mathbf{v}}_i)}$. Thus, $\mathbf{z}_i$ is equivalently the projection of the modality mean onto the per-instance dominant direction. Importantly, its magnitude retains the signed aggregation of the modality contributions: opposing contributions reduce the signal, while balanced opposition produces a zero signal. The signal strength admits the factorization $\left\|\boldsymbol z_i\right\|_2^2=\rho_i\,\delta_i$ and $\delta_i={\left(\boldsymbol r_{i,1}^{\top}\boldsymbol 1_k\right)^2}/{k}$, where $\rho_i$ measures spectral concentration and $\delta_i\in[0,1]$ measures net signed support. Therefore, a strong corroborated signal requires both a dominant shared component and limited cancellation among modality contributions. Proposition~\ref{prop:corroboration} derives this factorization and establishes the invariance of $\mathbf z_i$ to the joint sign ambiguity of the leading singular vectors.

\noindent\highlight{Instance-wise regularization.}
The singular value maximization objective and preceding construction characterize corroboration within an instance, but admit a degenerate solution if applied alone. All instances can collapse onto the same global leading direction, erasing cross-instance distinctions. To discourage this collapse, we regularize the corroborated signals across instances with an instance-wise contrastive term. For a batch of size $B$, we define:
\begin{equation}
\mathcal L_{\mathrm{IR}}
=
-\frac{1}{B}
\sum_{i=1}^{B}
\log
\frac{
\exp(\mathbf z_i^\top\mathbf z_i/\gamma)
}{
\sum_{j=1}^{B}
\exp(\mathbf z_i^\top\mathbf z_j/\gamma)
},
\label{eq:ir_loss}
\end{equation}
where $z_i$ corresponds to the corroborated signal of instance $i$. Minimizing $\mathcal{L}_{\mathrm{IR}}$ suppresses the similarity between the corroborated signals of different instances, preserving cross-instance
discrimination (see Proposition~\ref{prop:corroboration}) while retaining the within-instance spectral concentration
induced by $\mathcal L_{\mathrm{SVM}}$.

The task relevance of the corroborated signal is further induced
through joint optimization with the movement-prediction loss introduced
in Sec.~\ref{sec:robust}.
Specifically, \(\mathcal L_{\mathrm{SVM}}\) promotes the
dominant spectral component, while $\mathcal L_{\mathrm{cls}}$ supplies a task-alignment gradient for the net signed support along the leading component. Proposition~\ref{prop:gradient} further provides a local gradient analysis that clarifies how these complementary effects arise under joint optimization.

\subsection{Robust Prediction Layer}
\label{sec:robust}
The reliability of financial modalities varies considerably across trading days: news is often sparse, sentiment is noisy, and the market and technical series are prone to transient shocks~\cite{li2020multimodal,shapiro2022measuring,erdemlioglu2021market}. The corroborated latent asset signals should therefore survive the degradation of any single modality. We enforce this property by training \method{} to produce consistent outputs for a clean sample and a perturbed copy in which one modality is corrupted. For each mini-batch, we randomly select one modality $a$ from $\mathcal{M}$ and corrupt only that branch with a perturbation operator $\Pi$, instantiated as either modality masking or additive Gaussian noise: \begin{equation}
\widetilde{\mathbf{X}}^{a}=\Pi(\mathbf{X}^{a})=
\begin{cases}
\mathbf{0}, & \text{(Masking)},\\[2pt]
\mathbf{X}^{a}+\boldsymbol{\epsilon},\quad \boldsymbol{\epsilon}\sim\mathcal{N}(\mathbf{0},\sigma_\epsilon^2\mathbf{I}), & \text{(Gaussian noise)}.
\end{cases}
\end{equation}

The remaining modalities are left intact, $\widetilde{\mathbf{X}}^{m}=\mathbf{X}^{m}(m\neq a)$, which yields the corrupted multimodal input $\widetilde{\mathcal{X}}^{a}={\widetilde{\mathbf{X}}^{m}}({m\in\mathcal{M})}$.
The clean input $\mathcal{X}$ and its corrupted copy $\widetilde{\mathcal{X}}^{a}$ pass through the same encoders, alignment maps, and corroboration mining module (Sections~\ref{sec:encoder}--\ref{sec:svm}), producing task-conditioned corroborated signals $\mathbf{z}_i$ and $\tilde{\mathbf{z}}_i$ for each instance $i$. A shared prediction head $g(\cdot)$, implemented as a multilayer perceptron (MLP), maps each corroborated signal to a scalar movement logit, $\ell_i=g(\mathbf{z}_i), \tilde{\ell}_i=g(\tilde{\mathbf{z}}_i).$
The robustness loss penalizes the discrepancy between the clean and
perturbed views at both the decision and representation levels: 
\begin{equation}
\mathcal{L}_{\mathrm{rob}}=
\frac{1}{B}\sum_{i=1}^{B}
\Big[
\mathrm{MSE}(\ell_i,\tilde{\ell}_i)
+
\mathrm{MSE}(\mathbf{z}_i,\tilde{\mathbf{z}}i)
\Big].
\label{eq:rob}
\end{equation}
The first term stabilizes the predicted logit, while the second
preserves the corroborated signal under modality degradation. Because
$\mathbf{z}_i$ is sign-invariant, the representation-level loss is
unaffected by arbitrary SVD sign flips and encourages the model to rely on the remaining modalities.

\vspace{2mm}
\noindent\highlight{Optimization objectives.} 
The prediction objective provides task supervision for the
corroborated representation constructed in Sec.~\ref{sec:svm}.
For supervision, we apply a binary cross-entropy loss to the movement
logit predicted from the clean corroborated signal: $y_i\in\{0,1\}$, $ \mathcal{L}_{\mathrm{cls}}=
\frac{1}{B}\sum_{i=1}^{B}
\mathrm{BCE}(\ell_i,y_i).$
Combining the corroboration mining objectives with robustness and
supervision, the joint training objective of \method{} is
\begin{equation}
\mathcal L_{\method{}}
=
\mathcal L_{\mathrm{cls}}
+
\beta_1\mathcal L_{\mathrm{SVM}}
+
\beta_2\mathcal L_{\mathrm{IR}}
+
\beta_3\mathcal L_{\mathrm{rob}},
\label{eq:overall_objective}
\end{equation}
where $\beta_1$, $\beta_2$, and $\beta_3$ control the contributions of
singular value maximization, instance-wise regularization, and robust
consistency, respectively (see detailed algorithm flow in Appendix~\ref{app:algorithm_flow}).

\begin{table}[t]
\centering
\small
\caption{Per-mini-batch time complexity of the additional alignment components in \method{}.}
\vspace{-2mm}
\label{tab:complexity}
\begin{tabular}{@{}l|l@{}}
\hline
\textbf{Component} & \textbf{Time complexity} \\
\hline
\textit{Singular Value Maximization (Sec.~\ref{sec:svm})} & \\
\quad Gram-matrix construction   & $\mathcal{O}(B d_c k^2)$ \\
\quad Compact eigendecomposition & $\mathcal{O}(B k^3)$ \\
\quad Leading-direction recovery & $\mathcal{O}(B d_c k)$ \\
\hline
\textit{Instance-wise regularization (Sec.~\ref{sec:svm})} & \\
\quad Pairwise similarity        & $\mathcal{O}(B^2 d_c)$ \\
\hline
\textit{Robust consistency (Sec.~\ref{sec:robust})} & \\
\quad Perturbed-view recomputation & $\mathcal{O}(B d_c k^2)$ \\
\hline
Overall alignment cost           & $\mathcal{O}(B d_c k^2 + B k^3 + B^2 d_c)$ \\
Dominant term ($k\ll d_c$)       & $\mathcal{O}(B d_c k^2 + B^2 d_c)$ \\
\hline
\end{tabular}
\vspace{-2mm}
\end{table}

\subsection{Complexity Analysis}
Since $k\ll d_c$, we recover the leading singular triplet of
$\mathbf{V}_i\in\mathbb R^{d_c\times k}$ from the compact Gram matrix
$\mathbf{V}_i^\top\mathbf{V}_i\in\mathbb R^{k\times k}$, and then
construct the leading rank-1 component and corroborated signal
$\mathbf{z}_i$. As summarized in Table~\ref{tab:complexity}, the total
mini-batch overhead is $\mathcal{O}(B d_c k^2+B^2d_c)$, which is dominated by the modality encoders.

%% file: tables/multimodal_data.tex






 







\begin{table}[t]
\centering
\caption{Diverse modalities adopted for financial trading, including market, technical, news and sentiment streams.}
\vspace{-3mm}
\label{tab:data_modalities}

\footnotesize
\renewcommand{\arraystretch}{1.18}
\setlength{\tabcolsep}{4pt}

\begin{tabularx}{\linewidth}{
@{}
>{\raggedright\arraybackslash}p{0.16\linewidth}
>{\raggedright\arraybackslash}p{0.25\linewidth}
>{\raggedright\arraybackslash}X
@{}
}
\toprule
\textbf{Modality}
&
\textbf{Data Source}
&
\textbf{Feature}
\\
\midrule

\textbf{\textit{Market}}
&
Yahoo Finance API
&
Includes observations, namely open, high, low, adjusted close, and trading volume.
\\

\textbf{\textit{Technical}}
&
Finnhub API 
&
Provides price patterns, indicators, and valuation ratios, e.g., RSI, volatility, and P/E.
\\

\textbf{\textit{News}}
&
Alpaca API
&
Covers company events, announcements, and policy developments, etc.
\\

\textbf{\textit{Sentiment}}
&
Alpha Vantage API
&
Comprises investor reactions, analyst views, and social-media opinions, etc.
\\

\bottomrule
\end{tabularx}
\vspace{-4mm}
\end{table}




%% file: tables/baselines.tex

\providecommand{\best}[1]{\textbf{#1}}
\providecommand{\second}[1]{\underline{#1}}

\definecolor{oursblue}{HTML}{DDE7FA}

\begin{table*}[t]
\centering

\caption{
Overall performance across equity and cryptocurrency assets.
ARR\% denotes the annual rate of return, SR denotes the Sharpe ratio. Higher is better for both ARR\% and SR.
Best results are bolded and second-best results are underlined.}
\vspace{-3mm}
\label{tab:overall_performance_full}

\resizebox{\textwidth}{!}{
\begin{tabular}{
l*{12}{N}
}

\toprule

\multirow{2}{*}{\textbf{Method}}
& \multicolumn{2}{c}{\textbf{AAPL}}
& \multicolumn{2}{c}{\textbf{AMZN}}
& \multicolumn{2}{c}{\textbf{GOOG}}
& \multicolumn{2}{c}{\textbf{MSFT}}
& \multicolumn{2}{c}{\textbf{TSLA}}
& \multicolumn{2}{c}{\textbf{BTCUSD}}
\\

\cmidrule(lr){2-3}
\cmidrule(lr){4-5}
\cmidrule(lr){6-7}
\cmidrule(lr){8-9}
\cmidrule(lr){10-11}
\cmidrule(lr){12-13}

& \ARRhead & \SRhead
& \ARRhead & \SRhead
& \ARRhead & \SRhead
& \ARRhead & \SRhead
& \ARRhead & \SRhead
& \ARRhead & \SRhead
\\

\midrule


B\&H
& 28.48 & 0.62
& 29.42 & 0.65
& 107.62 & 2.02
& 70.66 & 1.78
& 140.97 & 1.39
& 53.48 & 1.42
\\

MACD
& -23.54 & -0.77
& -35.10 & -1.00
& 28.80 & 0.89
& 12.38 & 0.62
& 49.62 & 0.92
& 30.25 & 1.35
\\

ZMR
& 15.03 & 0.92
& -7.11 & -0.99
& 7.78 & 1.25
& 3.31 & 0.45
& 20.60 & 0.61
& 14.43 & 1.05
\\

SMA
& 16.87 & 0.83
& -9.03 & -0.50
& 74.02 & 2.14
& 3.26 & 0.63
& 25.11 & 0.69
& -9.93 & -0.60
\\

\midrule


LSTM
& 14.82 & 0.40
& 36.89 & 0.87
& 66.14 & 1.61
& 50.52 & 1.44
& 99.55 & 1.17
& 49.60 & 1.35
\\

Transformer
& 17.29 & 0.45
& 18.53 & 0.52
& 79.47 & 1.81
& 31.26 & 0.98
& 93.34 & 1.36
& 56.51 & 1.68
\\

DQN
& 27.43 & \second{0.93}
& 15.27 & 0.60
& -12.91 & -0.54
& 11.24 & 0.68
& 44.76 & 0.90
& 12.36 & 0.63
\\

PPO
& 21.65 & 0.58
& 25.67 & 0.64
& 40.95 & 1.09
& 32.73 & 1.32
& 38.59 & 0.88
& 26.98 & 0.93
\\

Kronos
& 38.38 & 0.91
& 40.23 & 0.96
& 63.28 & 1.67
& 37.07 & 1.41
& \second{142.71} & 1.63
& 52.61 & 1.91
\\

\midrule


Qwen3-8B
& 2.16 & 0.16
& 1.44 & 0.14
& 44.48 & 1.31
& 17.69 & 0.69
& 90.64 & 1.24
& 36.20 & 1.16
\\

DeepSeek-R1-0528
& 0.70 & 0.11
& 13.44 & 0.43
& 65.35 & 1.63
& 16.14 & 0.78
& 71.51 & 1.09
& 18.69 & 0.72
\\

Llama4-Scout-17B
& 18.50 & 0.48
& -21.49 & -0.42
& 75.70 & 1.82
& 29.28 & 1.11
& 108.10 & 1.39
& 39.92 & 1.22
\\

\midrule


FinAgent
& 29.85 & 0.89
& 24.84 & 0.58
& 109.08 & 2.04
& 69.18 & 2.20
& 123.85 & 1.30
& 24.45 & 0.84
\\

TradingAgents
& -11.30 & -0.15
& 22.30 & 0.64
& 63.72 & 1.53
& 46.18 & 1.54
& 57.56 & 1.01
& 45.24 & 1.48
\\

DeepFund
& 28.33 & 0.66
& 10.96 & 0.36
& 50.70 & 1.34
& \second{75.35} & 1.98
& 52.63 & \second{1.75}
& 9.17 & 0.41
\\

VTA
& \second{38.79} & 0.88
& \second{49.69} & \second{1.23}
& \second{110.73} & \second{2.22}
& 72.18 & \second{2.42}
& 121.84 & 1.50
& \second{68.84} & \second{1.92}
\\
\midrule
\rowcolor{oursblue}
\textbf{CoLAS~(ours)}
& \best{67.79} & \best{1.47}
& \best{56.98} & \best{1.25}
& \best{124.68} & \best{2.46}
& \best{97.16} & \best{2.73}
& \best{159.86} & \best{1.89}
& \best{84.64} & \best{2.65}
\\

\bottomrule

\end{tabular}
}
\end{table*}

%% file: tex/4experiment.tex
\section{EXPERIMENTS}
\label{sec:experiments}
In this section, we conduct experiments and analysis to answer the following research questions (\textbf{RQs}):
\begin{itemize}[leftmargin=*]
    \item \textbf{RQ1}: Does \method{} outperform other conventional baselines and state-of-the-art LLM trading agents? (Sec.~\ref{sec:overall_performance})
    \item \textbf{RQ2}: Does each component of \method{} contribute to performance? (Sec.~\ref{sec:ablation})
    \item \textbf{RQ3}: Can \method{} reduce the modality opposition to support corroboration enhancement? (Sec.~\ref{sec:corroboration_validation})
    \item \textbf{RQ4}: Is \method{} still effective on more metrics? (Sec.~\ref{sec:more_metrics})
    \item \textbf{RQ5}: Is \method{} stable under different model-specific hyperparameter
    settings and various window sizes? (Sec.~\ref{sec:hyperparameters})
\end{itemize}

\subsection{Experimental Setup}

\noindent\highlight{Datasets.}
We evaluate \method{} on six asset-specific datasets, including five widely traded U.S. equities: Apple Inc. (AAPL), Amazon.com Inc.
(AMZN), Alphabet Inc. (GOOG), Microsoft Corporation (MSFT), and
Tesla Inc. (TSLA), together with Bitcoin (BTCUSD). This selection aims
to showcase \method{}'s versatility and consistency across various
financial assets. These assets are chosen for their extensive news coverage and representation of different market sectors, and these data provide a robust
basis for assessing \method{}'s generalization capabilities across
diverse financial environments. 
Each dataset contains four temporally aligned modalities, as summarized in Table~\ref{tab:data_modalities}. The primary datasets are chronologically divided into non-overlapping training (2023-10-01--2024-09-30), validation (2024-10-01--2025-03-31), and testing
(2025-04-01--2025-09-30) periods. We further consider an extended evaluation period, in which the training and validation periods remain
unchanged and the test period spans 2025-04-01--2026-03-20. Under this extended test horizon setting, we include five additional U.S. equities spanning different sectors and firm-size groups, and evaluate four selected assets from the primary datasets in Appendices~\ref{app:broader} and~\ref{app:extended_test_horizon}.

\noindent\highlight{Evaluation metrics.}
We compare \method{} and the baselines using the financial metrics adopted in prior work~\cite{zhang2024multimodal,yu2024fincon}, including one profitability metric: annual rate of return (ARR), and one risk-adjusted performance metric: Sharpe ratio (SR). Definitions and formulas are as follows (more metrics are detailed in Appendix~\ref{app:evaluation_metrics}):
\begin{itemize}[leftmargin=*]
    \item \textit{Annual Rate of Return (ARR)} is the annualized average return rate, calculated as $\mathrm{ARR} = \frac{V_T - V_0}{V_0} \times \frac{C}{T}$, where $T$ is the total number of trading days and $C$ is the annualized factor, set to 252 for equities and 365 for cryptocurrency. $V_T$ and $V_0$ represent the final and initial portfolio values.
    \item \textit{Sharpe Ratio (SR)} measures the risk-adjusted returns of portfolios. It is defined as $\mathrm{SR} = \frac{\mathbb{E}[\mathbf{r}]}{\sigma[\mathbf{r}]}$, where $\mathbb{E}[\cdot]$ is the expectation, $\sigma[\cdot]$ is the standard deviation, and $\mathbf{r} = \left[\frac{V_1 - V_0}{V_0}, \frac{V_2 - V_1}{V_1}, \dots, \frac{V_T - V_{T-1}}{V_{T-1}}\right]^{\top}$ denotes the sequence of daily returns. 
\end{itemize}

\noindent\highlight{Baselines. }
We comprehensively evaluate \method{} against 16 diverse baselines categorized into Rule-based Strategies, Single-modal Models, Multimodal General LLMs and Multimodal Financial LLMs. First, the Rule-based Strategies include traditional quantitative methods: \textit{Buy-and-Hold (B\&H)}~\cite{li2012pamr}, \textit{MACD}~\cite{chong2008technical}, \textit{ZMR}~\cite{zhang2024multimodal}, and \textit{SMA}~\cite{brock1992simple}. Second, the Single-modal Models include \textit{LSTM}~\cite{yang2020qlib}, \textit{Transformer}~\cite{yang2020qlib}, \textit{DQN}~\cite{mnih2013playing}, \textit{PPO}~\cite{schulman2017proximal}, and \textit{Kronos}~\cite{shi2026kronos}. Third, the Multimodal General LLMs include \textit{Qwen3-8B}~\cite{yang2025qwen3}, \textit{DeepSeek-R1-0528}~\cite{guo2025deepseek,huang2025explainable}, and \textit{Llama4-Scout-17B}~\cite{meta2025llama4}. Fourth, the Multimodal Financial LLMs include \textit{FinAgent}~\cite{zhang2024multimodal}, \textit{TradingAgents}~\cite{xiao2024tradingagents}, \textit{DeepFund}~\cite{li2026time}, and \textit{VTA}~\cite{koa2025reasoning}. Further details are in Appendix~\ref{app:baselines}.

\noindent\highlight{Implementation details.}
\label{sec:implementation_details}
All experiments are conducted on a server equipped with 8 NVIDIA A100 GPUs. For supervised methods, we perform a hyperparameter search over learning rates in $\{2\times10^{-5}, 1\times10^{-4}, 1\times10^{-3}\}$ and look-back window sizes in $\{10, 14, 20\}$, where applicable. We fix the batch size to 64 and optimize supervised methods with AdamW~\cite{loshchilov2017decoupled}. All experiments are repeated five times, and we report the mean performance across the five runs. Statistical significance is assessed against the strongest baseline, with $p<0.05$ considered statistically significant.
More details on model architectures and hyperparameter settings are provided in Appendix~\ref{app:datasets} and~\ref{app:hyperparameter_settings}.

\begin{figure*}[t]
    \centering
    \includegraphics[width=0.98\textwidth]{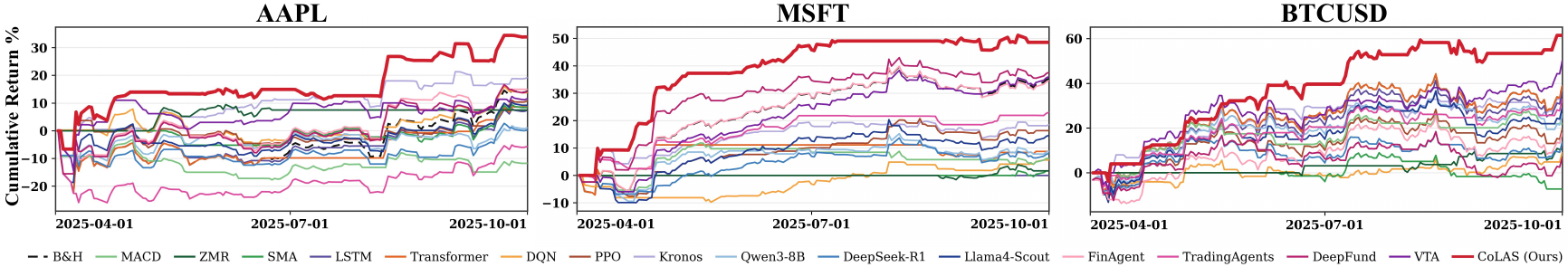}
    \vspace{-3mm}
    \caption{Overall performance over time between \method{} and other baselines across three assets from different sectors.}
    \vspace{-3mm}
    \label{fig:cr_curve}
\end{figure*} 

\subsection{Overall Performance (RQ1)}
\label{sec:overall_performance}
\noindent\highlight{Main results on the U.S. stock market. }
Table~\ref{tab:overall_performance_full} reports trading performance on five U.S. equities. We draw three observations. First, \method{} delivers the strongest profitability, reaching an ARR of 67.79\%, 56.98\%, 124.68\%, 97.16\%, and 159.86\% on the five stocks. Relative to the strongest baseline on each stock, \method{} improves ARR by 28.9\% on MSFT (from 75.35\% to 97.16\%) and by 12.0\% on TSLA (from 142.71\% to 159.86\%), and it surpasses the second-best method on every remaining stock. Second, \method{} also leads in risk-adjusted performance, obtaining the best SR on all five datasets, with particularly large margins on AAPL (1.47 \textit{vs.} 0.93) and MSFT (2.73 \textit{vs.} 2.42). The concurrent improvements in ARR and SR indicate that the return
gains are not accompanied by a deterioration in standard
risk-adjusted performance. Instead, \method{} improves the return-risk trade-off itself. Third, the improvement is consistent across baseline families. \method{} surpasses not only rule-based and reinforcement learning strategies but also the strongest financial LLM (VTA) and the time-series foundation model Kronos, even though these baselines already exploit multiple modalities or large-scale pretraining. 
This result highlights the effectiveness of explicitly modeling multimodal corroboration relative to the evaluated baselines, including models that rely on large-scale pretraining.

Figure~\ref{fig:cr_curve} plots the cumulative return over the backtesting horizon and shows how these gains accrue (see detailed cumulative return analyses on all assets in Appendix~\ref{app:cr_curve}). On AAPL, \method{} stays among the leaders from the start and advances in a sharp step in the later part of the horizon, finishing highest with a cumulative return of 33.9\%. On MSFT, the advantage builds earlier, as \method{} climbs steeply in the first months to a cumulative return of about 48.6\% and then holds this level to the end, whereas competing agents keep oscillating and finish lower. 
On both stocks, the stepwise cumulative-return curves indicate that a substantial portion of the gains is realized during a limited number of periods and then preserved over relatively stable intervals. 
Due to space limitations, we provide more detailed results in terms of a broader asset pool and an extended test horizon in Appendices~\ref{app:broader} and~\ref{app:extended_test_horizon}.

\noindent\highlight{Main results on cryptocurrency.}
Table~\ref{tab:overall_performance_full} shows that \method{} also performs strongly on BTCUSD,
achieving the highest ARR of 84.64\% and the highest SR of 2.65.
Relative to the runner-up VTA, CoLAS improves ARR by 23.0\%
and SR from 1.92 to 2.65. These results indicate that the advantage
of CoLAS extends from U.S. equities to the more volatile
cryptocurrency setting. Figure~\ref{fig:cr_curve} further shows that \method{} establishes an early lead on BTCUSD and remains above the competing methods for most of the
test period. This result is consistent with the effectiveness of the
learned corroboration-oriented representation under substantial
market volatility.

\subsection{Ablation Study (RQ2)}
\label{sec:ablation}
Table~\ref{tab:ablation_modules} reports leave-one-component-out ablations
on AAPL and BTCUSD. Removing any component degrades both ARR
and SR, indicating that SVM, IR, and RPL provide complementary
benefits.

\input{tables/RQ3_ablation_study}

\noindent
\highlight{Effect of singular value maximization (SVM).} 
Removing SVM causes the largest ARR reduction among the three ablations on both assets, with ARR decreasing from 67.79\% to 59.54\% on AAPL and from 84.64\% to 74.88\% on
BTCUSD. Its effect on SR is modest on AAPL, where SR decreases from 1.47 to 1.38, but is more pronounced on BTCUSD, where SR decreases from 2.65 to 2.17. These results suggest that dominant-component concentration primarily contributes to profitability while also improving risk-adjusted performance,
particularly on BTCUSD.

\noindent
\highlight{Effect of instance-wise regularization (IR).} 
Removing IR degrades both metrics on both assets. On AAPL, ARR decreases from 67.79\% to 60.49\%, while SR decreases
from 1.47 to 1.30. On BTCUSD, ARR decreases from 84.64\% to 76.23\%, and SR decreases from 2.65 to 2.33. These results are consistent with the role of IR in preserving discriminative corroborated signals across instances, which contributes to downstream trading performance.

\noindent
\highlight{Effect of robust prediction layer. (RPL)} 
RPL enforces consistency between the clean view and a perturbed
view in which one modality is corrupted. Removing RPL causes the largest SR reduction among the three ablations, from 1.47 to 1.29 on AAPL and from 2.65 to 2.14 on BTCUSD. In contrast, its removal causes the smallest ARR reduction on both assets. Together with the robustness study results in Appendix~\ref{app:robustness}, these findings show that RPL not only improves risk-adjusted performance under the clean-input setting but also reduces performance degradation when an input modality is unavailable or stochastically removed. 
Additional comparisons with generic fusion strategies and alternative corroborated-signal constructions are provided in Appendix~\ref{app:ablation} (Table~\ref{tab:ablation_study_2}).

\begin{figure}[t]
    \centering
    \vspace{-2mm}
    \includegraphics[width=0.87\linewidth]{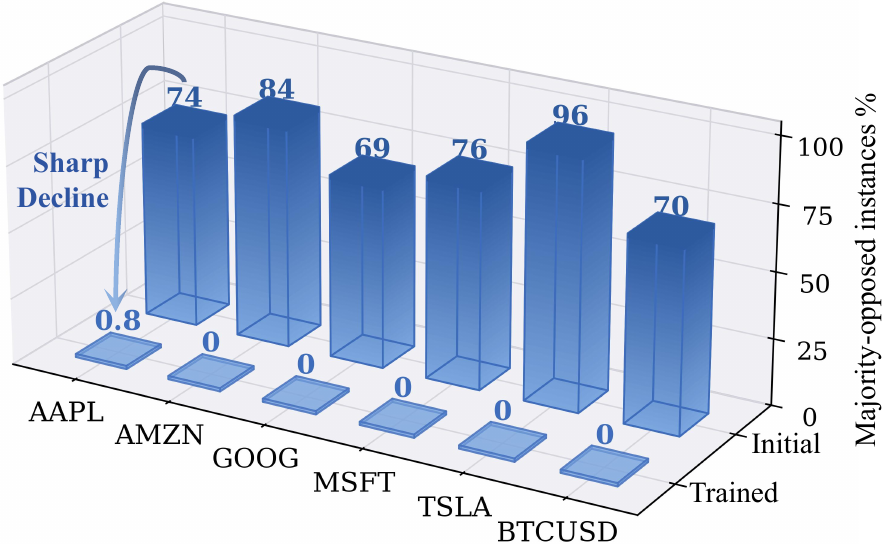}
    \vspace{-3mm}
    \caption{Prevalence of majority-opposed instances before and after joint optimization.}
    \vspace{-4mm}
    \label{fig:majority_opposed}
\end{figure}

\begin{figure}[t]
    \centering
    \includegraphics[width=\linewidth]{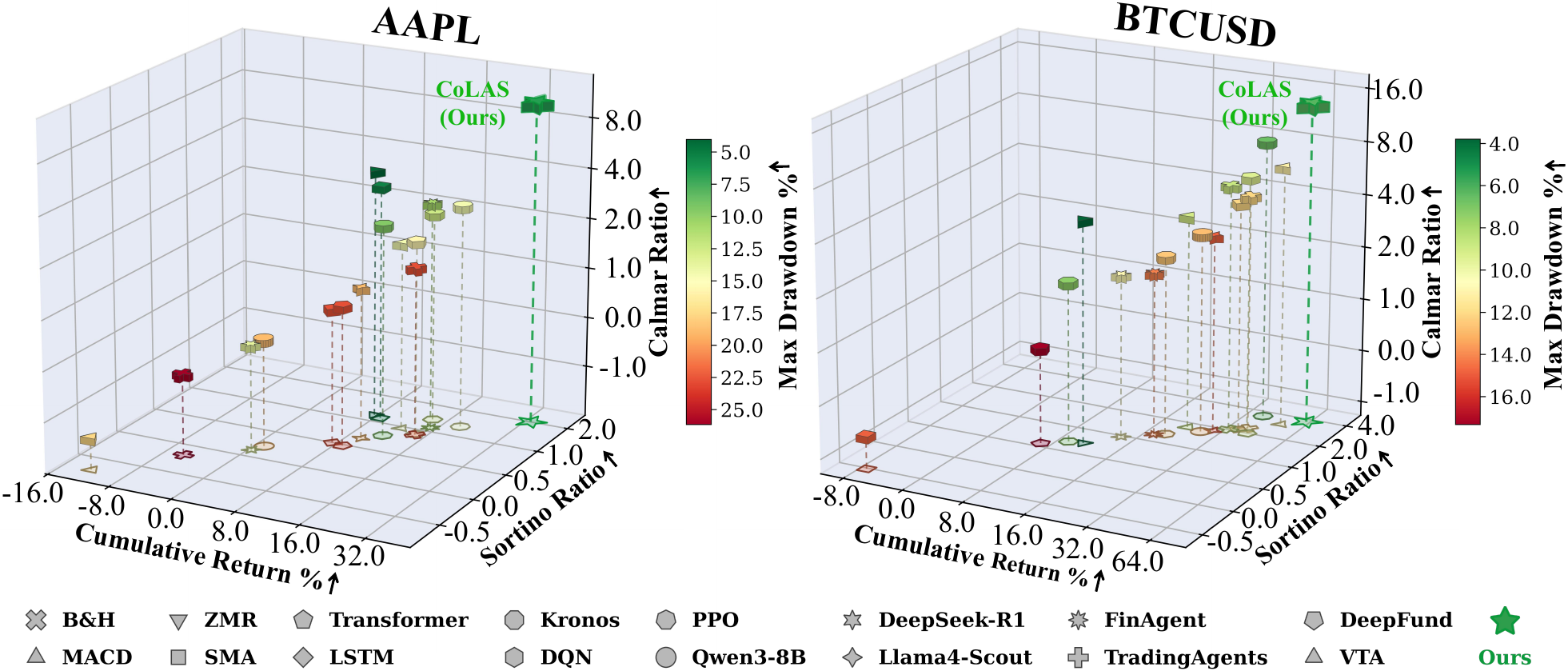}
    \vspace{-6mm}
    \caption{Extended metrics evaluated on different methods on the AAPL and BTCUSD datasets.}
    \Description{Extended metrics}
    \label{fig:exp_RQ4}
    \vspace{-4mm}
\end{figure}

\subsection{Corroboration Enhancement Analysis (RQ3)} 
\label{sec:corroboration_validation} 
To examine whether \method{} achieves the intended reduction in cross-modal opposition, we analyze the pairwise relationships among the four aligned modality representations.
According to our formulation, effective corroboration requires not only a dominant shared direction but also non-canceling modality contributions along this direction. Strong pairwise opposition may cause these contributions to cancel each other and consequently weaken the resulting corroborated signal. We therefore compute the pairwise cosine similarities among the four aligned modality representations and define an instance as majority-opposed when more than half of its modality pairs have negative cosine similarity. As shown in Figure~\ref{fig:majority_opposed}, before joint optimization, 69\%–96\% of the instances exhibit majority pairwise opposition across the six assets. After joint optimization, this proportion decreases to at most 0.8\%. This sharp decline shows that the optimization substantially reduces pairwise representational opposition. Because opposing modality directions can attenuate the signed aggregate used to construct the corroborated signal, the result supports the opposition-reduction aspect of the proposed corroboration mechanism.

\begin{figure*}[t]
    \centering
    \includegraphics[width=\textwidth]{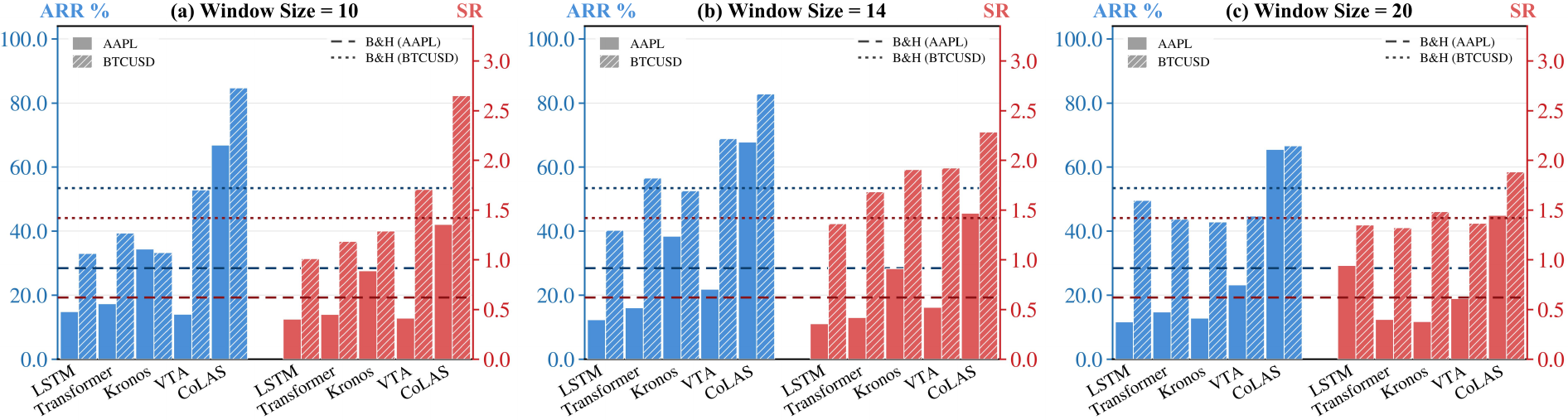}
    \vspace{-8mm}
    \caption{Sensitivity analyses of various window sizes on different methods on the AAPL and BTCUSD datasets.}
    \vspace{-3mm}
    \label{fig:exp_RQ5}
\end{figure*}

\begin{figure}[t]
    \centering
    \includegraphics[width=\linewidth]{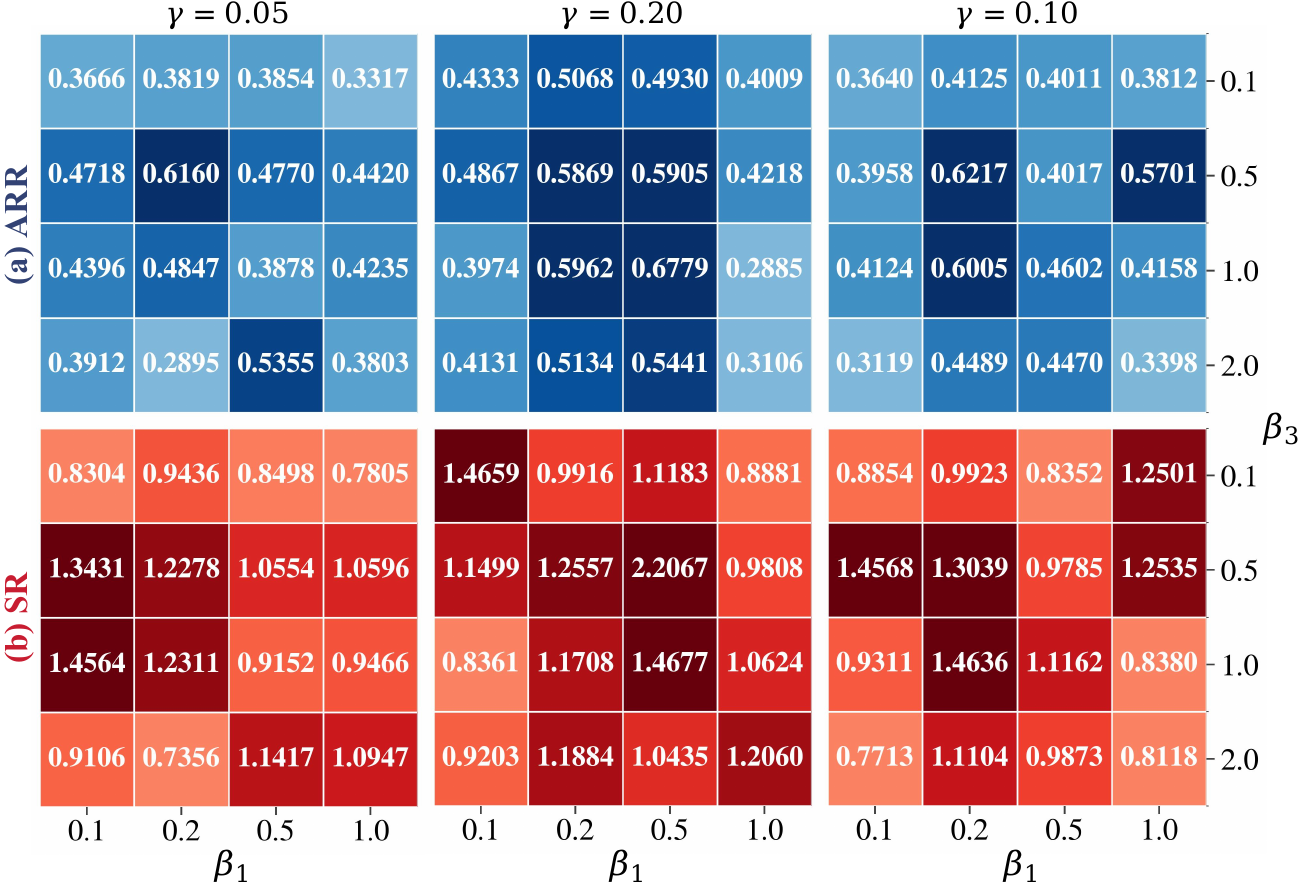}
    \vspace{-6mm}
    \caption{
    Hyperparameter sensitivity analysis of \method{} on AAPL
    in terms of ARR and SR. Corresponding results on BTCUSD are
    provided in Appendix~\ref{app:hyperparameter_analyses}.}
    \Description{Hyperparameter analysis}
    \label{fig:exp_RQ3_aapl}
    \vspace{-4mm}
\end{figure}

\subsection{Performance on More Metrics (RQ4)}
\label{sec:more_metrics}
Figure~\ref{fig:exp_RQ4} extends the evaluation to four additional
metrics: Cumulative Return (CR), Maximum Drawdown (MDD), Sortino
Ratio (SoR), and Calmar Ratio (CalR). These metrics complement ARR
and SR by assessing cumulative profitability, downside risk, and
drawdown-related performance. \method{} achieves favorable results
across these metrics on both AAPL and BTCUSD. On BTCUSD,
\method{} achieves the highest cumulative return, outperforming the
runner-up VTA by 23.0\%. It also limits the maximum drawdown to
6.27\% and obtains a Calmar ratio of 13.50, compared with 7.91 for
the runner-up. These results show that \method{} maintains a favorable
balance between profitability and downside risk across the evaluated
equity and cryptocurrency assets.

\subsection{Hyperparameter Analyses (RQ5)}
\label{sec:hyperparameters}
\highlight{Method-specific hyperparameters.} 
We analyze the sensitivity of \method{} to three method-specific hyperparameters: the singular value maximization weight $\beta_1$, the robustness weight $\beta_3$, and the separability temperature $\gamma$. Figure~\ref{fig:exp_RQ3_aapl} reports the results on AAPL, while the corresponding results on BTCUSD are provided in Appendix~\ref{app:hyperparameter_analyses}. From the AAPL results, we draw three observations.
\textit{\ding{172} \method{} is resilient across a wide span of settings.} Most configurations yield an ARR of roughly 40\%--68\% with SR within a narrow band of 0.9--1.5, indicating that a strong return--risk profile does not hinge on any finely tuned configuration.
\textit{\ding{173} Both loss weights favor intermediate values.} $\beta_1\in\{0.2, 0.5\}$ attains the best mean ARR (50.5\% and 48.3\%), while the extremes drop to about 40\%: a weak term leaves modalities loosely aligned, whereas a dominant one collapses them onto the shared direction at the expense of modality-specific signal. $\beta_3$ follows the same pattern, peaking at $\beta_3\in\{0.5, 1.0\}$ (ARR 50.7\% and 46.5\%; SR 1.27 and 1.12) and receding to roughly 41\% ARR at either extreme, since too little consistency regularization leaves the corroborated signal fragile under modality degradation while too much suppresses informative variation.
\textit{\ding{174} A balanced separability temperature serves best.} $\gamma=0.10$ yields the strongest mean ARR (47.9\%) and SR (1.19), ahead of the sharper $\gamma=0.05$ and softer $\gamma=0.20$, which over-separate or blur sample representations, respectively.
Overall, moderate weights paired with a balanced temperature give the most favorable ARR--SR trade-off, underscoring the practical usability of \method{}.

\noindent\highlight{Look-back window sizes.} We vary the look-back window size on the AAPL and BTCUSD datasets
in Figure~\ref{fig:exp_RQ5}. On AAPL, \method{} remains stable across the three evaluated window sizes, with ARR ranging from 65.49\% to 67.79\% and SR remaining around 1.45. The best performance is obtained at a window size of 14. Several baselines, most notably Kronos, exhibit greater sensitivity to the look-back
window, with its ARR ranging from 12.79\% to 38.38\%. In comparison, \method{} remains the best-performing method at
all three window sizes. On BTCUSD, \method{} also ranks first at every evaluated window size and achieves its best performance with the shortest window,
obtaining an ARR of 84.64\% and an SR of 2.65 at a window size of 10. Its performance then gradually decreases as the window size increases, suggesting that shorter historical windows are more
suitable for BTCUSD during the evaluated period. Overall, these results indicate that \method{} is relatively insensitive to the evaluated look-back window sizes and does not depend on a single carefully tuned window setting on either asset.

%% file: tables/RQ3_ablation_study.tex
\definecolor{oursblue}{HTML}{DDE7FA}
\begin{table}[t]
    \centering
    \caption{Ablation study of \method{} on AAPL and BTCUSD.
    Higher ARR and SR indicate better performance.}
    \vspace{-3mm}
    \label{tab:ablation_modules}

    \resizebox{0.85\linewidth}{!}{
        \begin{tabular}{ccc cc cc}
            \toprule

            \multicolumn{3}{c}{\textbf{Included Modules}}
            & \multicolumn{2}{c}{\textbf{AAPL}}
            & \multicolumn{2}{c}{\textbf{BTCUSD}} \\

            \cmidrule(lr){1-3}
            \cmidrule(lr){4-5}
            \cmidrule(lr){6-7}

            \textbf{SVM} & \textbf{IR} & \textbf{RPL}
            & \textbf{ARR\%$\uparrow$} & \textbf{SR$\uparrow$}
            & \textbf{ARR\%$\uparrow$} & \textbf{SR$\uparrow$} \\

            \midrule

            \cmark & \cmark & \xmark
            & \second{62.83} & 1.29
            & \second{79.81} & 2.14 \\

            \cmark & \xmark & \cmark
            & 60.49 & 1.30
            & 76.23 & \second{2.33} \\

            \xmark & \cmark & \cmark
            & 59.54 & \second{1.38}
            & 74.88 & 2.17 \\

            \rowcolor{oursblue}
            \cmark & \cmark & \cmark
            & \best{67.79} & \best{1.47}
            & \best{84.64} & \best{2.65} \\

            \specialrule{\heavyrulewidth}{0pt}{0pt}
        \end{tabular}
    }
    \vspace{-3mm}
\end{table}

%% file: tex/5conclusion.tex
\section{Conclusion}
\label{sec:conclusion}
In this paper, we proposed \method{}, a framework driven by explicit multimodal corroboration mining and robustness-aware consistency, jointly optimized for financial trading. 
\method{} addresses a limitation of existing approaches that aggregate multimodal information without explicitly characterizing signed cross-modal support at the instance level.
\method{} renews this and emphasizes that 
multimodal financial trading can benefit not only from incorporating complementary information, but also from explicitly identifying task-relevant, non-canceling support across heterogeneous financial modalities.
By concentrating shared structure, preserving signed modality support, and maintaining robustness under modality degradation, \method{} constructs a corroboration-oriented representation for financial prediction. Empirical results across diverse stock and cryptocurrency assets demonstrate that \method{} achieves the SOTA overall trade-off, effectively delivering profitable returns while robustly preserving the risk-adjusted performance.

%% file: tex/6appendix.tex
\section{Notation}
\label{app:notations}
To improve clarity and maintain consistent notation throughout the paper, we provide a notation table as a quick reference (see Table~\ref{tab:notation}) for the definitions of all symbols used.

\section{Theoretical Analysis}
\label{app:theoretical_analysis}
We analyze how CoLAS extracts task-conditioned multimodal
corroboration under joint optimization. In financial markets,
corroboration does not necessarily require unanimous agreement among
all modalities. Instead, it refers to a dominant shared component whose
signed modality contributions provide non-canceling support for the
financial prediction task. A modality may provide opposing evidence,
but such evidence reduces the strength of the shared financial signal
rather than being treated as equally supportive evidence.

For instance $i$, let
$\boldsymbol V_i=[\bar{\boldsymbol v}_i^1,\ldots,\bar{\boldsymbol v}_i^k]\in\mathbb R^{d_c\times k}$
collect the unit-normalized modality representations. Its compact SVD,
fixed spectral energy, leading component, signed modality contributions,
and corroborated signal are jointly defined as
\begin{equation}
\begin{aligned}
\boldsymbol V_i
&=\boldsymbol P_i\boldsymbol\Lambda_i\boldsymbol R_i^\top
&=\sum_{j=1}^{k}\lambda_{i,j}\,
\boldsymbol p_{i,j}\boldsymbol r_{i,j}^{\top},
\qquad
\lambda_{i,1}\geq\cdots\geq\lambda_{i,k}\geq0
\end{aligned}
\label{eq:app_svd}
\end{equation}
\begin{equation}
\begin{aligned}
\|\boldsymbol V_i\|_F^2
&=\operatorname{tr}\!\left(\boldsymbol V_i^\top\boldsymbol V_i\right)
=\sum_{m=1}^{k}\|\bar{\boldsymbol v}_i^m\|_2^2\\
&=\sum_{j=1}^{k}\lambda_{i,j}^2
=k
\end{aligned}
\label{eq:app_energy}
\end{equation}
\begin{equation}
\begin{gathered}
\boldsymbol V_i^{(1)}\coloneqq
\lambda_{i,1}\boldsymbol p_{i,1}\boldsymbol r_{i,1}^{\top},
\qquad
a_{i,m}\coloneqq
\boldsymbol p_{i,1}^{\top}\bar{\boldsymbol v}_i^m,\\
\boldsymbol z_i\coloneqq
\frac1k\boldsymbol V_i^{(1)}\boldsymbol 1_k
\end{gathered}
\label{eq:app_defs}
\end{equation}
where $\boldsymbol P_i=[\boldsymbol p_{i,1},\ldots,\boldsymbol p_{i,k}]$,
$\boldsymbol R_i=[\boldsymbol r_{i,1},\ldots,\boldsymbol r_{i,k}]$, and
$\boldsymbol 1_k\in\mathbb R^k$ is the all-one vector.

\begin{proposition}[\textbf{Corroboration}]
\label{prop:corroboration}
\label{prop:spectral_concentration}
The corroborated signal admits the closed form, strength factorization,
and residual-energy characterization
\begin{equation}
\begin{gathered}
\boldsymbol z_i
=\frac1k\!\left(\sum_{m=1}^{k}a_{i,m}\right)\!\boldsymbol p_{i,1},
\qquad
\|\boldsymbol z_i\|_2^2=\rho_i\,\delta_i,\\
\|\boldsymbol V_i-\boldsymbol V_i^{(1)}\|_F^2
=k-\lambda_{i,1}^2
\end{gathered}
\label{eq:app_corroborated_signal}
\end{equation}
where
\begin{equation}
\begin{aligned}
\rho_i
&\coloneqq\frac{\lambda_{i,1}^2}{k}
\in\!\left[\tfrac1k,1\right]
&&\text{(spectral concentration)},\\
\delta_i
&\coloneqq
\frac{\left(\boldsymbol r_{i,1}^{\top}\boldsymbol 1_k\right)^2}{k}
=\cos^2\angle\!\left(\boldsymbol r_{i,1},\boldsymbol 1_k\right)
\in[0,1]
&&\text{(net signed support)}
\end{aligned}
\label{eq:app_strength_factorization}
\end{equation}
Fix $\gamma_1>0$. The feasible spectral set and spectral objective are
\begin{equation}
\begin{aligned}
\Omega\coloneqq\Bigl\{
&(\lambda_{i,1},\ldots,\lambda_{i,k})\in\mathbb R_+^k:\\[-2pt]
&\lambda_{i,1}\geq\max_{j>1}\lambda_{i,j},
\quad
\sum_{j=1}^{k}\lambda_{i,j}^2=k
\Bigr\},
\end{aligned}
\end{equation}
\begin{equation}
\mathcal L_{\mathrm{SVM},i}
=-\log
\frac{\exp(\lambda_{i,1}/\gamma_1)}
{\sum_{j=1}^{k}\exp(\lambda_{i,j}/\gamma_1)}.
\end{equation}
Over $\Omega$, $\mathcal L_{\mathrm{SVM},i}$ has the unique global minimizer
\begin{equation}
(\lambda_{i,1},\lambda_{i,2},\ldots,\lambda_{i,k})
=(\sqrt{k},0,\ldots,0).
\label{eq:app_rank_one_optimum}
\end{equation}
Moreover, transferring spectral energy from any non-leading component
to the leading component strictly decreases
$\mathcal L_{\mathrm{SVM},i}$ and the residual energy
$\sum_{j=2}^{k}\lambda_{i,j}^2$.
\end{proposition}

\begin{proof}
By the Eckart--Young theorem, orthonormality of
$\boldsymbol P_i,\boldsymbol R_i$, Eq.~\eqref{eq:app_energy},
$\|\boldsymbol p_{i,1}\|_2=\|\boldsymbol r_{i,1}\|_2=1$, and
$\boldsymbol V_i^\top\boldsymbol p_{i,1}
=\lambda_{i,1}\boldsymbol r_{i,1}$,
\begin{equation}
\begin{aligned}
\|\boldsymbol V_i-\boldsymbol V_i^{(1)}\|_F^2
&=\sum_{j=2}^{k}\lambda_{i,j}^2
=\sum_{j=1}^{k}\lambda_{i,j}^2-\lambda_{i,1}^2\\
&=k-\lambda_{i,1}^2,
\end{aligned}
\end{equation}
\begin{equation}
\begin{aligned}
\boldsymbol z_i
&=\frac1k\boldsymbol V_i^{(1)}\boldsymbol 1_k
=\frac{\lambda_{i,1}}{k}
\boldsymbol p_{i,1}\boldsymbol r_{i,1}^{\top}\boldsymbol 1_k\\
&=\frac{\lambda_{i,1}
\left(\boldsymbol r_{i,1}^{\top}\boldsymbol 1_k\right)}{k}
\boldsymbol p_{i,1}
=\frac1k\!\left(\sum_{m=1}^{k}a_{i,m}\right)\!\boldsymbol p_{i,1},
\end{aligned}
\end{equation}
\begin{equation}
\begin{aligned}
\|\boldsymbol z_i\|_2^2
&=\frac{\lambda_{i,1}^2
\left(\boldsymbol r_{i,1}^{\top}\boldsymbol 1_k\right)^2}{k^2}\\
&=\underbrace{\frac{\lambda_{i,1}^2}{k}}_{\rho_i}
\cdot
\underbrace{
\frac{\left(\boldsymbol r_{i,1}^{\top}\boldsymbol 1_k\right)^2}{k}
}_{\delta_i},
\end{aligned}
\end{equation}
\begin{equation}
\begin{aligned}
\sum_{m=1}^{k}a_{i,m}
&=\boldsymbol 1_k^\top\boldsymbol V_i^\top\boldsymbol p_{i,1}
=\lambda_{i,1}\boldsymbol r_{i,1}^{\top}\boldsymbol 1_k,\\
\delta_i
&=\frac{\left(\boldsymbol r_{i,1}^{\top}\boldsymbol 1_k\right)^2}
{\|\boldsymbol r_{i,1}\|_2^2\|\boldsymbol 1_k\|_2^2}
\in[0,1],\\
1\leq\lambda_{i,1}^2
&\leq\sum_{j=1}^{k}\lambda_{i,j}^2=k
\quad\Longrightarrow\quad
\rho_i\in\!\left[\tfrac1k,1\right].
\end{aligned}
\end{equation}

Suppose $\lambda_{i,j}>0$ for some $j>1$. For
$0<\tau<\lambda_{i,j}^2$, define the energy-preserving path and the
corresponding softmax probabilities as
\begin{equation}
\begin{gathered}
\lambda_{i,1}(\tau)=\sqrt{\lambda_{i,1}^2+\tau},
\qquad
\lambda_{i,j}(\tau)=\sqrt{\lambda_{i,j}^2-\tau},\\
\sum_{\ell=1}^{k}\lambda_{i,\ell}(\tau)^2=k
\end{gathered}
\end{equation}
\begin{equation}
q_{i,\ell}(\tau)
=\frac{\exp(\lambda_{i,\ell}(\tau)/\gamma_1)}
{\sum_{r=1}^{k}\exp(\lambda_{i,r}(\tau)/\gamma_1)}.
\end{equation}
\begin{equation}
\begin{aligned}
\frac{\partial\mathcal L_{\mathrm{SVM},i}}
{\partial\lambda_{i,1}}
&=\frac{q_{i,1}-1}{\gamma_1},
&
\frac{\partial\mathcal L_{\mathrm{SVM},i}}
{\partial\lambda_{i,j}}
&=\frac{q_{i,j}}{\gamma_1},\\[2pt]
\frac{\mathrm d\mathcal L_{\mathrm{SVM},i}}{\mathrm d\tau}
&=\frac{q_{i,1}(\tau)-1}
{2\gamma_1\lambda_{i,1}(\tau)}
-\frac{q_{i,j}(\tau)}
{2\gamma_1\lambda_{i,j}(\tau)}
<0,
\end{aligned}
\end{equation}
\begin{equation}
\sum_{\ell=2}^{k}\lambda_{i,\ell}(\tau)^2
=\sum_{\ell=2}^{k}\lambda_{i,\ell}^2-\tau.
\end{equation}
The ordering of the non-leading singular values may be restored by
permutation, which does not change the objective. Therefore, any feasible
spectrum with a positive non-leading singular value cannot be optimal.
At the optimum,
$\lambda_{i,j}=0$ for every $j>1$, and the fixed-energy condition gives
$\lambda_{i,1}=\sqrt{k}$, proving the uniqueness of
Eq.~\eqref{eq:app_rank_one_optimum}.
\end{proof}

\noindent\textbf{\textit{Remark.}}
$\boldsymbol z_i$ is invariant under
$(\boldsymbol p_{i,1},\boldsymbol r_{i,1})
\mapsto(-\boldsymbol p_{i,1},-\boldsymbol r_{i,1})$
and equals $\operatorname{proj}_{\boldsymbol p_{i,1}}(\bar{\boldsymbol v}_i)$ for the modality
mean $\bar{\boldsymbol v}_i=\frac1k\sum_{m}\bar{\boldsymbol v}_i^m$. Spectral concentration
alone does not imply corroboration: under rank-1 structure ($\rho_i=1$),
$[\boldsymbol u,\boldsymbol u,\boldsymbol u,-\boldsymbol u]
\Rightarrow\boldsymbol z_i=\tfrac12\boldsymbol u$
while
$[\boldsymbol u,\boldsymbol u,-\boldsymbol u,-\boldsymbol u]
\Rightarrow\boldsymbol z_i=\boldsymbol 0$,
so partial conflict is admitted ($\delta_i>0$) but balanced opposition ($\delta_i=0$)
yields no corroborated signal. The proposition characterizes the inductive preference of
the spectral objective over the feasible singular-value space; it does not claim that
finite-step optimization must attain an exactly rank-1 representation. Instead, the
objective favors a larger leading-energy proportion
$\rho_i=\lambda_{i,1}^2/k$ and a smaller residual
$\|\mathbf V_i-\mathbf V_i^{(1)}\|_F^2$. The distinction of CoLAS does not arise from
rank-1 projection alone, but from jointly learning the per-instance spectral concentration,
retaining signed cross-modal cancellation, aligning the net support with the prediction task,
and enforcing consistency under modality degradation.

\input{tables/notation}

\noindent\begin{proposition}[\textbf{Local gradient contributions under joint optimization}]
\label{prop:gradient}
Consider the joint objective
$\mathcal L_{\mathrm{CoLAS}}
=\mathcal L_{\mathrm{cls}}
+\beta_1\mathcal L_{\mathrm{SVM}}
+\beta_2\mathcal L_{\mathrm{IR}}
+\beta_3\mathcal L_{\mathrm{rob}}$,
where $\beta_1>0$. We isolate the local gradient contributions of
$\beta_1\mathcal L_{\mathrm{SVM}}$ and
$\mathcal L_{\mathrm{cls}}$ within this objective.

For instance $i$, define
\begin{equation}
\begin{aligned}
q_{i,j}
&\coloneqq
\frac{\exp(\lambda_{i,j}/\gamma_1)}
{\sum_{\ell=1}^{k}\exp(\lambda_{i,\ell}/\gamma_1)},
&
\mathcal L_{\mathrm{SVM},i}
&\coloneqq-\log q_{i,1},
\end{aligned}
\label{eq:app_svm_probability}
\end{equation}
\begin{equation}
\begin{aligned}
t_i&\coloneqq2y_i-1,
&
s_i&\coloneqq\sum_{m=1}^{k}a_{i,m},\\
\ell_i&\coloneqq
\frac{s_i}{k}\mathbf w_i^\top\mathbf p_{i,1}+b_i
\end{aligned}
\label{eq:app_local_logit}
\end{equation}
\begin{equation}
\mathcal L_{\mathrm{cls},i}
\coloneqq\log\!\left(1+\exp(-t_i\ell_i)\right).
\label{eq:app_local_cls}
\end{equation}
\begin{equation}
\label{eq:app_task_aligned_support}
A_i
\coloneqq t_i\!\left(\mathbf w_i^\top\mathbf p_{i,1}\right)s_i.
\end{equation}
Here, the prediction head is locally linearized around $\mathbf z_i$ as
$g(\mathbf z)\approx\mathbf w_i^\top\mathbf z+b_i$. The weighted spectral term
and the classification term satisfy
\begin{equation}
\begin{aligned}
\frac{\partial\!\left(\beta_1\mathcal L_{\mathrm{SVM},i}\right)}
{\partial\lambda_{i,1}}
&=\frac{\beta_1(q_{i,1}-1)}{\gamma_1}<0,\\
\frac{\partial\!\left(\beta_1\mathcal L_{\mathrm{SVM},i}\right)}
{\partial\lambda_{i,j}}
&=\frac{\beta_1q_{i,j}}{\gamma_1}>0,
\qquad j>1,
\end{aligned}
\label{eq:app_svm_gradient}
\end{equation}
\begin{equation}
\frac{\partial\mathcal L_{\mathrm{cls},i}}{\partial a_{i,m}}
=-\frac{t_i\sigma(-t_i\ell_i)}{k}
\mathbf w_i^\top\mathbf p_{i,1}.
\label{eq:app_cls_gradient}
\end{equation}
Holding $\mathbf p_{i,1}$ and $\mathbf w_i$ fixed and treating
$\{a_{i,m}\}_{m=1}^{k}$ as local coordinates, a local gradient step of size
$\eta>0$ induced only by $\mathcal L_{\mathrm{cls},i}$ satisfies
\begin{equation}
\begin{aligned}
A_i^{+}-A_i
&=\eta\,\sigma(-t_i\ell_i)
\left(\mathbf w_i^\top\mathbf p_{i,1}\right)^2\\
&\geq0,
\end{aligned}
\label{eq:app_task_support_update}
\end{equation}
with strict inequality whenever
$\mathbf w_i^\top\mathbf p_{i,1}\neq0$.
\end{proposition}

\begin{proof}
Writing the weighted spectral term explicitly and differentiating gives
\begin{equation}
\begin{aligned}
\beta_1\mathcal L_{\mathrm{SVM},i}
&=-\frac{\beta_1\lambda_{i,1}}{\gamma_1}
+\beta_1\log\sum_{\ell=1}^{k}
\exp(\lambda_{i,\ell}/\gamma_1),\\[2pt]
\frac{\partial\!\left(\beta_1\mathcal L_{\mathrm{SVM},i}\right)}
{\partial\lambda_{i,1}}
&=-\frac{\beta_1}{\gamma_1}
+\frac{\beta_1q_{i,1}}{\gamma_1}
=\frac{\beta_1(q_{i,1}-1)}{\gamma_1}<0,\\[2pt]
\frac{\partial\!\left(\beta_1\mathcal L_{\mathrm{SVM},i}\right)}
{\partial\lambda_{i,j}}
&=\frac{\beta_1q_{i,j}}{\gamma_1}>0,
\qquad j>1,
\end{aligned}
\end{equation}
where $0<q_{i,1}<1$ and $q_{i,j}>0$.

For the classification term, the chain rule and the induced local update yield
\begin{equation}
\begin{aligned}
\frac{\partial\mathcal L_{\mathrm{cls},i}}{\partial\ell_i}
&=-t_i\sigma(-t_i\ell_i),
&
\frac{\partial\ell_i}{\partial a_{i,m}}
&=\frac1k\mathbf w_i^\top\mathbf p_{i,1},\\[2pt]
\frac{\partial\mathcal L_{\mathrm{cls},i}}{\partial a_{i,m}}
&=-\frac{t_i\sigma(-t_i\ell_i)}{k}
\mathbf w_i^\top\mathbf p_{i,1},
\end{aligned}
\end{equation}
\begin{equation}
\begin{aligned}
s_i^{+}
&=\sum_{m=1}^{k}\left(
 a_{i,m}-\eta
 \frac{\partial\mathcal L_{\mathrm{cls},i}}{\partial a_{i,m}}
 \right)\\
&=s_i+\eta\,t_i\sigma(-t_i\ell_i)
\mathbf w_i^\top\mathbf p_{i,1},\\[2pt]
A_i^{+}-A_i
&=t_i\!\left(\mathbf w_i^\top\mathbf p_{i,1}\right)
\!\left(s_i^{+}-s_i\right)\\
&=\eta\,t_i^2\sigma(-t_i\ell_i)
\left(\mathbf w_i^\top\mathbf p_{i,1}\right)^2\\
&=\eta\,\sigma(-t_i\ell_i)
\left(\mathbf w_i^\top\mathbf p_{i,1}\right)^2
\geq0,
\end{aligned}
\end{equation}
where the last equality uses $t_i^2=1$.
\end{proof}

\noindent\textbf{\textit{Remark.}}
The proposition isolates two complementary local gradient contributions
within the joint objective. The weighted spectral term
$\beta_1\mathcal L_{\mathrm{SVM}}$ favors a more prominent leading
spectral component, whereas the classification term supplies a
non-negative task-alignment update for the aggregate signed support
$A_i$. This update does not require every individual modality
contribution $a_{i,m}$ to have the same sign.

Moreover, $A_i$ is invariant under the joint SVD sign transformation
$(\mathbf p_{i,1},\mathbf r_{i,1})
\mapsto
(-\mathbf p_{i,1},-\mathbf r_{i,1})$, because both
$\mathbf w_i^\top\mathbf p_{i,1}$ and
$s_i=\sum_m a_{i,m}$ change sign simultaneously. The complete parameter
update additionally contains the gradients of
$\mathcal L_{\mathrm{IR}}$ and $\mathcal L_{\mathrm{rob}}$ and updates the
encoders, alignment maps, singular directions, and prediction head
jointly. Therefore, the analysis characterizes the local inductive
contributions of the two examined objectives rather than guaranteeing
that every complete joint update monotonically increases
$\lambda_{i,1}$ or $A_i$, or establishing global convergence.

\input{tables/news_coverage}

\section{Implementation Details}
\label{app:implementation}

\subsection{Algorithm Flow}
\label{app:algorithm_flow}
To facilitate reproducibility and provide a clearer view of the overall procedure, we summarize the workflow of \method{} in Algorithm~\ref{alg:colas}. Each training iteration proceeds in four steps. Step~\ding{192} encodes every modality of a sampled instance and maps it to the alignment space, where the normalized vectors are stacked into the per-instance matrix $\mathbf{V}_i$. Step~\ding{193} applies SVD to $\mathbf{V}_i$, extracts its leading rank-1 component, and constructs the task-conditioned corroborated signal $\mathbf{z}_i$ from the signed
modality contributions retained in this component. It then computes the corroboration mining losses $\mathcal L_{\mathrm{SVM}}$ and
$\mathcal L_{\mathrm{IR}}$. Step~\ding{194} randomly perturbs one modality to form a corrupted view, recomputes $\tilde{\mathbf{z}}_i$ and the two logits, and yields the robustness loss $\mathcal{L}_{\mathrm{rob}}$ together with the classification loss $\mathcal{L}_{\mathrm{cls}}$. Step~\ding{195} combines the four losses and updates the parameters by gradient descent.

\begin{algorithm}[t]
\caption{\method{} Training}
\label{alg:colas}
\begin{algorithmic}[1]
\Require Dataset $\mathcal{X}=\{(\{\mathbf{X}^{m}_i\}_{m\in\mathcal{M}},\,y_i)\}_{i=1}^{N}$, $k=|\mathcal{M}|$ modalities;
         encoders $f_m(\cdot;\boldsymbol{\theta}_m)$, alignment maps $\psi_m$, head $g$;
         weights $\beta_1,\beta_2,\beta_3$, noise level $\sigma_\epsilon$.
\Ensure Optimized parameters $\Theta$.
\For{each training iteration}
  \State Sample a batch $\{(\{\mathbf{X}^{m}_i\}_{m\in\mathcal{M}},\,y_i)\}_{i=1}^{B}$.
  \Statex \textbf{\ding{192} Encode and align} ($\forall m\in\mathcal{M}$):
  \State $\bar{\mathbf{v}}^{m}_i=\psi_m\!\big(f_m(\mathbf{X}^{m}_i;\boldsymbol{\theta}_m)\big)$, $\ell_2$-normalized.
  \State Stack $\mathbf{V}_i=[\bar{\mathbf{v}}^{m_1}_i,\ldots,\bar{\mathbf{v}}^{m_k}_i]\in\mathbb{R}^{d_c\times k}$.
  \Statex \textbf{\ding{193} Mine corroboration:}
  \State SVD $\mathbf{V}_i=\mathbf{P}_i\mathbf{\boldsymbol{\Lambda}}_i\mathbf{R}_i^{\top}$; leading direction $\mathbf{z}_i=\frac{1}{k}\left(\sum_{m=1}^{k}a_{i,m}\right)\mathbf{p}_{i,1}$.
  \State Losses: $\mathcal{L}_{\mathrm{SVM}}$ (Eq.~\ref{eq:svm_loss}), $\mathcal{L}_{\mathrm{IR}}$ (Eq.~\ref{eq:ir_loss}).
  \Statex \textbf{\ding{194} Enforce robustness:}
  \State Perturb one modality $a$: $\widetilde{\mathbf{X}}^{a}_i=\Pi(\mathbf{X}^{a}_i)$, others unchanged.
  \State Recompute $\tilde{\mathbf{z}}_i$ via \ding{192}--\ding{193}; logits $\ell_i=g(\mathbf{z}_i)$, $\tilde{\ell}_i=g(\tilde{\mathbf{z}}_i)$.
  \State Losses: $\mathcal{L}_{\mathrm{rob}}$ (Eq.~\ref{eq:rob}), $\mathcal{L}_{\mathrm{cls}}=\tfrac{1}{B}\sum_i\mathrm{BCE}(\ell_i,y_i)$.
  \Statex \textbf{\ding{195} Update:}
  \State $\mathcal{L}=\mathcal{L}_{\mathrm{cls}}+\beta_1\mathcal{L}_{\mathrm{SVM}}+\beta_2\mathcal{L}_{\mathrm{IR}}+\beta_3\mathcal{L}_{\mathrm{rob}}$;\quad $\Theta\leftarrow\Theta-\eta\nabla_{\Theta}\mathcal{L}$.
\EndFor
\State \Return $\Theta$.
\end{algorithmic}
\end{algorithm}

\subsection{Datasets and Processing}
\label{app:datasets}
We construct six primary asset-specific datasets (see Table~\ref{tab:dataset_statistics}), including five
widely traded U.S. equities and one cryptocurrency. For each asset,
we collect four forms of financial information, namely market,
technical, news, and sentiment data. Although these modalities may
share upstream information, they differ in their information content,
construction procedures, update frequencies, and noise
characteristics. All modalities are aligned at the asset-day level according to their
observation or release timestamps, with only information available
before the prediction cutoff on day $t$ used to predict the asset
movement on day $t+1$. We describe the assets and the construction of each
modality below:

\paragraph{Asset.}
The primary asset set comprises Apple Inc. (AAPL, Information
Technology, mega-cap), Amazon.com Inc. (AMZN, Consumer
Discretionary, mega-cap), Alphabet Inc. (GOOG, Communication
Services, mega-cap), Microsoft Corporation (MSFT, Information
Technology, mega-cap), Tesla Inc. (TSLA, Consumer Discretionary,
mega-cap), and Bitcoin (BTCUSD, Cryptocurrency). We additionally
include Coca-Cola Company (KO, Consumer Defensive, mega-cap),
Johnson \& Johnson (JNJ, Healthcare, mega-cap), United Parcel
Service (UPS, Industrials, large-cap), Upwork (UPWK, Communication
Services, small-cap), and Sprouts Farmers Market (SFM, Consumer
Defensive, mid-cap). The additional equities cover different sectors
and firm-size groups, thereby broadening the empirical asset coverage
beyond the predominantly mega-cap companies in the primary set.

The primary datasets are divided chronologically into non-overla\-pping
training (2023-10-01--2024-09-30), validation
(2024-10-01--2025-03-31), and testing
(2025-04-01--2025-09-30) periods. We further construct an
extended-horizon setting with the same training and validation
periods and a test period spanning 2025-04-01--2026-03-20. This
setting includes the five additional equities and four selected assets
from the primary set.

\paragraph{Market.}
We collect daily market observations from Yahoo Finance. For each
asset and trading day, the market modality contains the open price,
high price, low price, adjusted close price, and trading volume. The
adjusted close accounts for applicable corporate actions, whereas
the remaining fields retain the values reported by the data provider.
Each numerical feature is standardized using statistics estimated
from the training period, and the same transformation is applied to
the validation and test periods.

\paragraph{Technical.}
The technical modality contains indicators computed from historical
market observations and valuation variables retrieved from Finnhub. The market-based
indicators characterize historical trend, momentum, and volatility,
including the relative strength index (RSI), moving average
convergence divergence (MAC\-D), rate of change (ROC), simple and
exponential moving averages (SMA and EMA), Bollinger Bands,
average true range (ATR), and historical volatility. The valuation
variables include the price-to-earnings ratio (P/E) and
price-to-book ratio (P/B). Although the market-based indicators are
derived from historical price and volume observations, the valuation
variables contain firm-level information that cannot be directly
recovered from daily OHLCV data. The valuation variables are used in a point-in-time manner. For each trading day, we retain only the latest value whose publication or effective timestamp is no later than the corresponding information cutoff. No current valuation snapshot is backfilled into earlier dates. 

For both the market and technical modalities, we use a two-layer LSTM encoder with a hidden dimension of 768 and a dropout rate of 0.1.

\input{tables/news_prompt}

\paragraph{News.}
We retrieve time-stamped, asset-related news from the Alpaca API.
The collected articles cover company-specific developments,
material corporate events, product announcements, earnings-related
information, and relevant policy or macroeconomic developments.
For each asset and date, the available articles are ordered by their
release timestamps and organized into a fixed news-only prompt, as
shown in Table~\ref{tab:news_agent_prompt_example}. Each prompt is
encoded using Qwen2.5-7B-Instruct, released in September 2024 with
a knowledge cutoff in June 2024.

\input{tables/sentiment_prompt}

\paragraph{Sentiment.}
Sentiment-related information is collected from the Alpha Vantage
API. This modality reflects investor reactions, analyst views, and
market opinions concerning the target asset. Observations associated
with the same asset and date are grouped and organized into a fixed
sentiment-only prompt, as shown in Table~\ref{tab:sentiment_agent_prompt_example}, which is encoded using DeepSeek-R1, released
in January 2025 with a knowledge cutoff in July 2024. In contrast to
the news modality, which describes reported events and factual
developments, the sentiment modality characterizes how market
participants interpret or react to the available information.

For both modalities, we extract the hidden states from the last
transformer layer and use the hidden state corresponding to the final
non-padding token as the daily representation. These representations
are computed once offline before downstream model training and remain
fixed across all experimental runs. Within each modality, the same
prompt template is used for all assets and dates. The prompts contain
only the corresponding modality-specific records available by the
given date and exclude other modalities, movement labels, future
information, validation results, and test-period outcomes.

\paragraph{Temporal Alignment and Information Availability.}
All timestamped textual records are converted to the asset-specific reference timezone before temporal aggregation.
We use Eastern Time for U.S. equities and UTC for BTCUSD.
For U.S. equities, news and sentiment records released after the official market close or on a non-trading day are assigned to the next trading day.
For BTCUSD, records are grouped according to UTC calendar days.
For a sample indexed by trading day $t$, CoLAS uses only modality observations assigned to days $t-T, \dots, t-1$.
Specifically, all records assigned to day $t$, including pre-market and intraday textual information, are excluded.
Market-derived technical indicators are computed recursively using price and volume observations no later than day $t-1$.
This conservative alignment ensures that every input record is available before the corresponding trading decision is made.

\subsection{Baselines}
\label{app:baselines}
We comprehensively evaluate \method{} against 16 diverse baselines categorized into Rule-based Strategies, Single-modal Models, Multimodal General LLMs, and Multimodal Financial LLMs. First, the Rule-based Strategies include traditional quantitative methods: \textbf{Buy-and-Hold (B\&H)}, \textbf{MACD}, \textbf{ZMR}, and \textbf{SMA}. Second, the Single-modal Models include \textbf{LSTM}, \textbf{Transformer}, \textbf{DQN}, \textbf{PPO}, and \textbf{Kronos}. Third, the Multimodal General LLMs include \textbf{Qwen3-8B}, \textbf{DeepSeek-R1-0528}, and \textbf{Llama4-Scout-17B}. Furthermore, the Multimodal Financial LLMs include \textbf{FinAgent}, \textbf{TradingAgents}, \textbf{DeepFund}, and \textbf{VTA}. Below we briefly describe each baseline:

\begin{itemize}[leftmargin=*]
    \item \noindent\textbf{Rule-based Strategies.} We include traditional quantitative strategies. \textbf{Buy-and-Hold (B\&H)} holds the asset throughout the evaluation horizon, ignoring short-term fluctuations. \textbf{MACD} generates signals from MACD--signal line crossovers to capture trend momentum. \textbf{Z-score Mean Reversion (ZMR)} assumes prices revert to a statistical mean, with entries and exits determined by Z-score thresholds. \textbf{SMA} trades on moving-average levels and crossovers to reflect smoothed trend direction.
    
    \item \noindent\textbf{Single-modal Models.} These models rely on price and technical signals alone.
    \textbf{LSTM}~\cite{yang2020qlib} models sequential dependencies with gated memory for next-step price forecasting.
    \textbf{Transformer}~\cite{yang2020qlib} leverages self-attention to capture long-range temporal interactions for price prediction.
    \textbf{DQN}~\cite{mnih2013playing} approximates the action-value function with deep networks to select trades from market states.
    \textbf{PPO}~\cite{schulman2017proximal} optimizes a clipped surrogate objective to update policies stably and sample-efficiently.
    \textbf{Kronos}~\cite{shi2026kronos} is a decoder-only foundation model pretrained on large-scale candlestick (OHLCV) data from global exchanges, forecasting future price movements autoregressively.

    \item \noindent\textbf{Multimodal General LLMs.} We include strong general-purpose LLMs as text-driven trading agents.
    \textbf{Qwen3-8B}~\cite{yang2025qwen3}, \textbf{DeepSeek-R1-0528}~\cite{guo2025deepseek,huang2025explainable}, and \textbf{Llama4-Scout-17B}~\cite{meta2025llama4} are prompted to infer directional signals from both textual and numeric inputs (e.g., news summaries and recent prices) and to produce trading decisions under a unified prompting and execution protocol. We use the exact released checkpoints specified in our code without task-specific fine-tuning. All three models receive the same prompt template and the same temporally aligned textual and numerical inputs, including news summaries, sentiment summaries, recent market observations, and technical indicators. We use a temperature of 0, a top-p value of 1.0. Each model is instructed to return one final decision, either BUY or SELL.

    \item \noindent\textbf{Multimodal Financial LLMs.} We compare against multimodal methods specialized for financial trading (introduced in chronological order).
    \textbf{FinAgent}~\cite{zhang2024multimodal} is a finance-oriented LLM agent that performs market reasoning and decision making with domain-specific instructions and tools.
    \textbf{TradingAgents}~\cite{xiao2024tradingagents} is a multi-agent LLM framework for automated trading decision making.
    \textbf{DeepFund}~\cite{li2026time} is a finance-focused LLM agent that integrates market information and textual evidence for fund-style trading decisions.
    \textbf{VTA}~\cite{koa2025reasoning} translates historical price series into verbal technical analysis and optimizes the natural-language reasoning with an inverse-MSE reward, then conditions a time-series backbone on this reasoning to produce interpretable forecasts. All multimodal baselines are constructed from the same temporally aligned source records and are evaluated using the same test period, transaction costs, and long-or-flat Qlib protocol. Hyperparameters that require model selection are selected exclusively on the validation period and are fixed before testing.
\end{itemize}

\subsection{Evaluation Metrics}
\label{app:evaluation_metrics}

We evaluate trading performance using six financial metrics: Annual Rate of Return (ARR), Cumulative Return (CR), Sharpe Ratio (SR), Maximum Drawdown (MDD), Calmar Ratio (CalR), and Sortino Ratio (SoR). ARR and CR measure profitability, SR, CalR, and SoR measure risk-adjusted profitability, and MDD measures downside risk. Let $V_0$ and $V_T$ denote the initial and final portfolio values over a trading period of $T$ days. Let $C=252$ denote the number of trading days in one year, for cryptocurrency we use $C=365$. The daily return sequence is defined as
\begin{equation}
    \mathbf r =
    \left[
    \frac{V_1 - V_0}{V_0},
    \frac{V_2 - V_1}{V_1},
    \ldots,
    \frac{V_T - V_{T-1}}{V_{T-1}}
    \right]^\top .
\end{equation}

\begin{itemize}[leftmargin=*]

    \item \textbf{Annual Rate of Return (ARR).}
    ARR measures the annualized average return over the evaluation period:
    \begin{equation}
        \mathrm{ARR}
        =
        \frac{V_T - V_0}{V_0}
        \times
        \frac{C}{T}.
    \end{equation}

    \item \textbf{Cumulative Return (CR).}
    CR measures the total portfolio return over the full evaluation period:
    \begin{equation}
        \mathrm{CR}
        =
        \frac{V_T - V_0}{V_0}.
    \end{equation}
    Thus, ARR is the annualized form of CR: $\mathrm{ARR}
        =
        \mathrm{CR}
        \times
        \frac{C}{T}$.

    \item \textbf{Sharpe Ratio (SR).}
    SR measures risk-adjusted return using the standard deviation of returns:
    \begin{equation}
        \mathrm{SR}
        =
        \frac{\mathbb E[\mathbf r]}{\sigma[\mathbf r]},
    \end{equation}
    where $\mathbb E[\cdot]$ denotes expectation, $\sigma[\cdot]$ denotes standard deviation.

    \item \textbf{Maximum Drawdown (MDD).}
    MDD measures the largest peak-to-trough loss during the evaluation period:
    \begin{equation}
    \small
        \mathrm{MDD}
        =
        \max_{0 \le t \le T}
        \frac{P_t - V_t}{P_t},
        \quad
        P_t = \max_{0 \le s \le t} V_s .
    \end{equation}
    A lower MDD indicates smaller downside risk.

    \item \textbf{Calmar Ratio (CalR).}
    CalR compares annualized return with maximum drawdown:
    \begin{equation}
        \mathrm{CalR}
        =
        \frac{\mathrm{ARR}}{\mathrm{MDD}}.
    \end{equation}
    A higher CalR indicates stronger return relative to downside risk.

    \item \textbf{Sortino Ratio (SoR).}
    SoR measures risk-adjusted return while focusing only on downside volatility:
    \begin{equation}
        \mathrm{SoR}
        =
        \frac{\mathbb E[\mathbf r]}{\sigma[\mathbf r^{-}]},
    \end{equation}
    where $\mathbf r^{-}=\{r_t \in \mathbf r \mid r_t < 0\}$ denotes the negative-return subsequence. A higher SoR indicates better downside-risk-adjusted performance.

\end{itemize}

\subsection{Backtesting Protocol}
\label{app:backtest}

The \textit{Backtesting} module maps the binary directional
predictions produced by \method{} to target portfolio states and
evaluates the resulting trading strategy using the Qlib backtesting
engine. For each trading day, the model produces either a positive
or non-positive directional prediction, which corresponds to a long
or cash target state, respectively. Qlib then converts consecutive
target states into concrete portfolio actions, including entering a
long position, liquidating an existing position, maintaining the
current position, or remaining in cash. The resulting strategy is
evaluated using cumulative return, annualized return, Sharpe ratio,
maximum drawdown, and other standard financial metrics.

\noindent\highlight{Binary price-direction labels and HOLD actions.}
For an instance indexed by trading day $t$, the multimodal input is
constructed using observations available up to day $t-1$, and the
prediction target describes the subsequent close-to-close price
movement from day $t$ to day $t+1$. Specifically,
\begin{equation}
y_t=
\mathbb{I}\left[p_{t+1}^c\geq p_t^c\right],
\label{eq:yt}
\end{equation}
where $p_t^c$ denotes the adjusted closing price on trading day $t$.
Accordingly, $y_t=1$ represents a non-negative price movement from
$p_t^c$ to $p_{t+1}^c$, whereas $y_t=0$ represents a negative price
movement. An unchanged adjusted closing price is therefore included
in the positive class by definition.

The supervised learning task remains binary and does not introduce
a separate HOLD class. Instead, the binary prediction
$\widehat{y}_t$ specifies the target portfolio state:
$\widehat{y}_t=1$ corresponds to a long position, whereas
$\widehat{y}_t=0$ corresponds to a cash position. In this
long-or-flat setting, a non-positive prediction indicates a cash
target rather than a short position. The notions of BUY, SELL, and HOLD used in the backtesting protocol
refer to portfolio execution actions rather than prediction classes.
In particular, HOLD arises when the predicted target state remains
unchanged across two consecutive trading days. The complete state
transition logic used by Qlib is summarized in
Table~\ref{tab:trading_actions}.

\begin{table}[t]
\centering
\caption{Mapping from consecutive predicted target states to
portfolio execution actions in the Qlib backtesting engine.}
\label{tab:trading_actions}
\begin{tabular}{ccc}
\toprule
Day $t$
& Day $t+1$
& Execution action \\
\midrule
BUY  & BUY  & Hold the existing long position \\
BUY  & SELL & Liquidate the long position \\
SELL & BUY  & Enter a long position \\
SELL & SELL & Remain in cash \\
\bottomrule
\end{tabular}
\end{table}

Table~\ref{tab:trading_actions} describes only the transition from
consecutive predicted target states to portfolio actions. It neither
introduces a third prediction class nor changes the temporal
definition of the supervised label in Eq.~\ref{eq:yt}. For example,
two consecutive BUY states indicate that the existing long position
is maintained, whereas two consecutive SELL states indicate that
the portfolio remains in cash.

\input{tables/broader_asset_pool}
\input{tables/extended_test_horizons}

\noindent\highlight{Execution timing and close-to-close returns.}
For the sample indexed by day $t$, model inference is performed after the information cutoff of day $t-1$ and before the close of day $t$.
Because all observations assigned to day $t$ are excluded, the prediction $\hat{y}_t$ uses no information from day $t$ or later.
The resulting target portfolio state is submitted before the close of day $t$ and is assumed to be executed at the daily closing price $p_t^c$ under the Qlib daily-frequency backtesting convention.
The prediction $\hat{y}_t$ estimates the direction of the close-to-close price movement from $p_t^c$ to $p_{t+1}^c$. Consequently, when the model
changes its target state from cash to long on day $t$, the position
is entered at $p_t^c$, and the subsequent one-day holding return is
determined by $r_{t+1}
=
\frac{p_{t+1}^c-p_t^c}{p_t^c}$
before transaction costs. Thus, the close-to-close interval evaluated
for prediction and trading is $(t,t+1)$ rather than $(t-1,t)$. This daily-frequency setting follows Qlib's close-price execution convention and does not explicitly model intraday order latency or closing-auction slippage.

\noindent\highlight{Assumptions used in the backtesting setup.}
All methods are evaluated under the same daily, low-frequency,
long-or-flat protocol implemented using Qlib. Fractional shares are
allowed, so portfolio allocation is not affected by integer-share
rounding. We set \texttt{risk\_degree=0.99}, meaning that at most
$99\%$ of the current portfolio value is allocated to the long
position, while the remaining $1\%$ is retained as a cash buffer.
Accordingly, the residual cash position is induced by the prescribed
risk degree rather than by integer-share constraints.

When the predicted target state changes from positive to
non-positive, the existing long position is fully liquidated. No
short selling or borrowing is involved. We set the Qlib transaction
costs to \texttt{open\_cost=0.0005} and
\texttt{close\_cost=0.0005}, corresponding to $5$ basis points per
side and $10$ basis points for a complete entry--exit cycle. The
same cost configuration is applied to all evaluated methods. We do not impose an explicit turnover constraint and do not
separately model slippage or market impact, following the shared
experimental settings adopted in prior work
\cite{zhang2024multimodal,yu2024fincon}. These factors are not the
primary focus of this study, which examines predictive performance
under a common and consistent backtesting engine. Their potential
influence is partly mitigated by the daily trading frequency and the
use of highly liquid assets.

\subsection{Hyperparameter Settings}
\label{app:hyperparameter_settings}
We utilize AdamW~\cite{loshchilov2017decoupled} as the optimizer, keeping its default momentum coefficients $(0.9, 0.999)$. The weight decay is set to 0.01, the gradient norm is clipped to $1.0$, and the batch size is fixed to 64. For trainable methods, we search the learning rate over $\{2\times10^{-5}, 1\times10^{-4}, 1\times10^{-3}\}$ and the look-back window over $\{10, 14, 20\}$. For our method, we perform a grid search for the alignment weight $\beta_1$ over $\{0.1, 0.2, 0.5, 1.0\}$ and the robustness weight $\beta_3$ over $\{0.1,0.5, 1.0, 2.0\}$, and fix the instance-wise regularization weight $\beta_2$ to $0.1$. We also fix the temperature $\gamma_1$ to $0.05$ and search $\gamma_2$ over $\{0.05, 0.1, 0.2\}$. To ensure a fair comparison, all reinforcement-learning baselines are trained and evaluated in the same environment and under the same backtesting protocol.

\section{Additional Results}
\label{app:additional}
We provide additional experiments and analyses that further examine \method{}, covering a broader asset pool of 11 assets across different sectors, extended test horizons with the training and validation periods held fixed, a hyperparameter analysis on the cryptocurrency asset, and backtest equity curves on six assets. To keep the broader-asset and extended-horizon evaluations computationally tractable, we report representative and competitive baselines from each methodological family, selected based on their primary-experiment performance rather than the outcomes of these additional evaluations.

\subsection{Broader Asset Pool}
\label{app:broader}
To examine whether \method{} generalizes beyond the primary asset
set, we extend the evaluation to five additional U.S. equities
spanning different sectors, volatility profiles, and market-cap
ranges: the Coca-Cola Company (KO, Consumer Defensive, mega-cap),
Johnson \& Johnson (JNJ, Healthcare, mega-cap), United Parcel
Service (UPS, Industrials, large-cap), Upwork (UPWK, Communication
Services, small-cap), and Sprouts Farmers Market (SFM, Consumer
Defensive, mid-cap). Results are reported in
Table~\ref{tab:broader_asset_pool}.

As shown in Table~\ref{tab:broader_asset_pool}, \method{} achieves
the highest CR, ARR, and SR on all five additional assets. The
improvement is relatively modest on JNJ, where B\&H remains
competitive in CR and ARR, but is more pronounced on UPS and UPWK.
On SFM, \method{} is the only method that produces positive
cumulative and annualized returns, whereas all evaluated baselines
record negative returns. These results suggest that the performance
advantage is not limited to the technology-oriented assets in the
primary evaluation, although its magnitude varies across assets. The risk-adjusted results show a similar pattern. \method{} ranks
first in SoR and CalR on all five assets and achieves the lowest MDD
on UPS, UPWK, and SFM. On KO, VTA obtains a lower MDD than
\method{} (6.42\% versus 7.52\%), while on JNJ, TradingAgents achieves the lowest MDD, followed jointly by PPO and CoLAS, both with an MDD of 6.32\%. Therefore, the broader-asset evaluation indicates consistent improvements in
returns and risk-adjusted performance, but \method{} does not
uniformly achieve the lowest drawdown on every asset.

\subsection{Extended Test Horizon}
\label{app:extended_test_horizon}
We further evaluate whether the advantage of \method{} persists under a longer and more recent test horizon. Specifically, we extend the test period to 2025-04-01 through 2026-03-20, keeping the original training and validation sets unchanged, and re-evaluate on AAPL, GOOG, TSLA, and BTCUSD. This setting tests whether the model maintains its performance as market conditions drift away from the original evaluation window. Results are reported in Table~\ref{tab:extended_test_horizon}.

As shown in Table~\ref{tab:extended_test_horizon}, \method{}
achieves the highest CR, ARR, and SR on all four assets over the
extended horizon. The improvement is relatively moderate on GOOG,
where FinAgent remains competitive in return and Sharpe ratio, but
is more evident on AAPL and BTCUSD. In particular, on BTCUSD,
\method{} obtains a CR of 46.71\% and an SR of 1.62, while the
buy-and-hold strategy records a negative return and TradingAgents,
the second-best method, achieves a CR of 30.90\% and an SR of 1.00.
On TSLA, \method{} also improves over the strongest baseline across
all three return-oriented metrics. The risk-oriented results are broadly consistent with the return
metrics. \method{} achieves the highest SoR and CalR on all four
assets and the lowest MDD on AAPL, TSLA, and BTCUSD. On GOOG, PPO
produces the lowest MDD of 7.12\%, while \method{} ranks second with
8.08\%; nevertheless, \method{} retains the highest SoR and CalR on
this asset. Overall, \method{} ranks first on 23 of the 24 reported
asset--metric combinations and second on the remaining one. These
results provide evidence that its advantage persists over the
examined extended horizon, although the magnitude of improvement
varies across assets and the evaluation remains limited to the four
selected markets.

\subsection{Comparison of Corroborated Signal Constructions}
\label{app:ablation}
Table~\ref{tab:ablation_study_2} evaluates whether the gains of \method{} depend on the proposed corroborated-signal construction rather than generic multimodal aggregation. We compare \method{} with mean pooling, concatenation, and two construction variants, \method{}$_{FSC}$ and \method{}$_{unsigned}$. The full model consistently achieves the best ARR and SR on both assets, supporting the importance of combining a dominant shared component with signed modality support.

To isolate the effect of signal construction, all variants use the
same modality representations and downstream prediction setting,
while replacing only the representation passed to the prediction
layer. Let
$\mathbf{V}_i=[\bar{\mathbf{v}}_i^1,\ldots,
\bar{\mathbf{v}}_i^k]$ denote the matrix of normalized modality
representations, with
$\mathbf{V}_i=\mathbf{P}_i\mathbf{\Lambda}_i\mathbf{R}_i^\top$.
Mean Pooling directly averages the modality representations,
\[
\mathbf{z}_i^{\mathrm{mean}}
=
\frac{1}{k}\sum_{m=1}^{k}\bar{\mathbf{v}}_i^m
\]
whereas Concatenation stacks them and applies a learnable projection,
\[
\mathbf{z}_i^{\mathrm{cat}}
=
\phi_{\mathrm{cat}}
\left(
[\bar{\mathbf{v}}_i^1;\ldots;\bar{\mathbf{v}}_i^k]
\right)
\]
These two constructions retain the available modality information
but do not explicitly identify a shared spectral direction.

The first-singular-component variant,
$\method{}_{\mathrm{FSC}}$, uses only the leading component,
\[
\mathbf{z}_i^{\mathrm{FSC}}
=
\lambda_{i,1}\mathbf{p}_{i,1}
\]
and therefore tests whether the dominant spectral component alone
is sufficient. The unsigned variant retains the leading direction
but replaces the signed modality contributions
$a_{i,m}=\mathbf{p}_{i,1}^{\top}\bar{\mathbf{v}}_i^m$
with their magnitudes:
\[
\mathbf{z}_i^{\mathrm{unsigned}}
=
\frac{1}{k}
\left(
\sum_{m=1}^{k}|a_{i,m}|
\right)
\mathbf{p}_{i,1}
\]
Under this construction, modalities pointing in opposite directions
are treated as equally supportive and no longer cancel. In contrast,
the full model constructs
\[
\mathbf{z}_i
=
\frac{1}{k}
\left(
\sum_{m=1}^{k}a_{i,m}
\right)
\mathbf{p}_{i,1}
\]
so its magnitude depends on both the dominant shared direction and
the net signed support along that direction.

The results distinguish the contribution of these two factors.
Mean Pooling and Concatenation produce substantially lower ARR than
all three spectral constructions on both assets, indicating that
generic aggregation does not reproduce the performance of the
corroborated-signal formulation. Retaining only the first singular
component improves ARR to 37.46\% on AAPL and 60.24\% on BTCUSD,
but its AAPL SR remains 0.77, below the 0.90 obtained by
Concatenation. Thus, extracting a dominant component alone does not
uniformly improve risk-adjusted performance. Ignoring contribution
signs also reduces performance: $\method{}_{\mathrm{unsigned}}$
reaches an ARR of 41.23\% and 58.42\% on AAPL and BTCUSD,
respectively, compared with 67.79\% and 84.64\% for the full model.
Relative to the strongest alternative construction, the full model
improves ARR by 26.56 percentage points and SR by 0.27 on AAPL,
and improves ARR by 24.40 percentage points and SR by 1.02 on
BTCUSD. These comparisons suggest that the dominant spectral
component and cancellation-aware signed aggregation make
complementary contributions under the two evaluated settings.

\begin{figure}[t]
    \centering
    \includegraphics[width=0.48\textwidth]{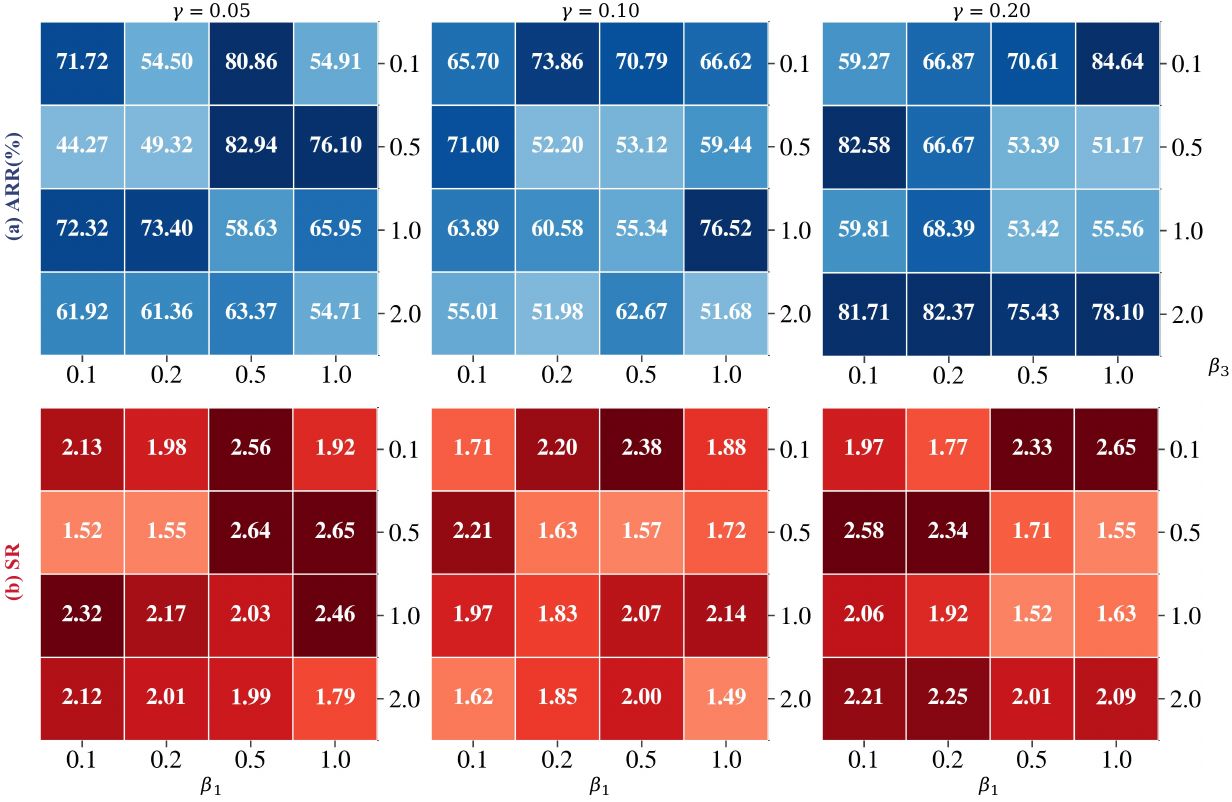}
    \caption{Hyperparameter sensitivity analysis of \method{} on BTCUSD with ARR(\%) and SR. We vary the temperature $\gamma$, the alignment weight \(\beta_{1}\), and the robustness weight \(\beta_{3}\).}
    \Description{Hyperparameter analysis}
    \label{fig:exp_RQ3_btcusd}
    \Description{Hyperparameter analysis}
\end{figure}

\input{tables/app_ablation_study}

\subsection{Hyperparameter Analyses}
\label{app:hyperparameter_analyses}
We report several findings from Figure~\ref{fig:exp_RQ3_btcusd}.
\textbf{First, \method{} remains robust on BTCUSD across all tested configurations.}
Every setting achieves an ARR above 44\% and an SR above 1.49, indicating that its strong return–risk profile does not depend on careful hyperparameter tuning.
\textbf{Second, the consensus weight $\beta_1$ has little effect on aggregate performance.}
The mean ARR stays within a narrow 63.5\%--65.8\% band across all four values ($65.8$, $63.5$, $65.0$, and $64.6$ for $\beta_1=0.1,0.2,0.5,1.0$), and the mean SR likewise varies only between 1.96 and 2.07.
Unlike on AAPL, where an intermediate $\beta_1$ was clearly preferable, on BTCUSD the performance is relatively insensitive to $\beta_1$.
\textbf{Third, the robustness weight $\beta_3$ favors smaller values.}
The lowest setting $\beta_3=0.1$ attains the best mean ARR (68.4\%) and SR (2.12), and both metrics decline mildly as $\beta_3$ grows, dropping to 61.9\% ARR at $\beta_3=0.5$ and to 1.95 SR at $\beta_3=2.0$.
Because cryptocurrency modalities are already highly volatile, only a light perturbation-invariance term is needed, and a heavier one excessively smooths informative modality-specific variation.
\textbf{Finally, the separability temperature $\gamma$ shows no monotone trend, and moderate sharpness is not required.}
The middle value $\gamma=0.10$ is in fact a mild trough (61.9\% ARR, 1.89 SR), whereas the sharper $\gamma=0.05$ (64.1\%, 2.12) and the softer $\gamma=0.20$ (68.1\%, 2.04) are slightly stronger.
Overall, across the entire grid \method{} sustains an ARR above 44\% and an SR above 1.49 on BTCUSD, and its aggregate performance shifts by only a few points as the hyperparameters vary, underscoring that its strong return--risk trade-off is not the product of delicate tuning but a stable property of the method.

\begin{figure}[t]
    \centering
    \includegraphics[width=0.48\textwidth]{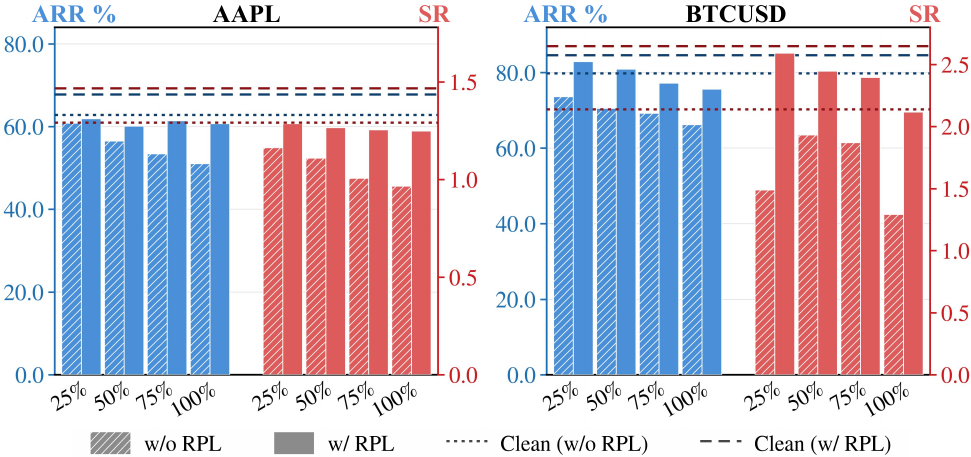}

\includegraphics[width=0.48\textwidth]{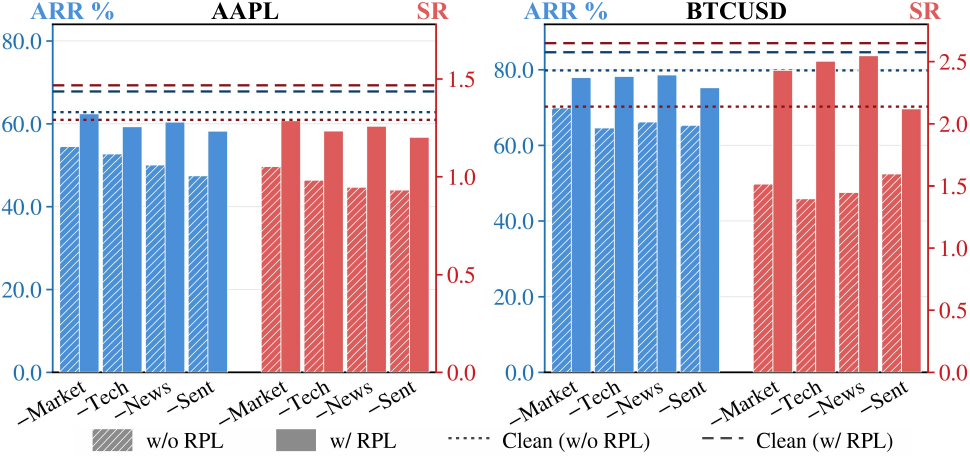}
    \caption{Stochastic modality and single-modality removal dropout on AAPL and BTCUSD.}
    \label{robustness_radom_drop}
    \Description{robust2}
\end{figure}

\subsection{Robustness Study}
\label{app:robustness}
Figures~\ref{robustness_radom_drop}
evaluate \method{} and its variant
without RPL under deterministic single-modality removal and
stochastic modality dropout. This experiment complements the
clean-setting ablation in Table~\ref{tab:ablation_modules} by
directly examining whether RPL reduces performance degradation when
input modalities become unavailable.

In the deterministic setting, one selected modality is zeroed out
for every test instance, while the other three modalities remain
unchanged. We separately remove the market, technical, news, and
sentiment modalities, yielding four fixed missing-modality
conditions. In the stochastic setting, each test instance is assigned
one modality before evaluation, and the assigned modality is zeroed
out independently with probability
$p\in\{0.25,0.50,0.75,1.00\}$. Thus, $p=1.00$ means that every
instance loses one modality, although the missing modality may differ
across instances. The dotted and dashed horizontal lines denote the
corresponding clean performance without and with RPL, respectively.

Under deterministic removal, both variants perform below their clean
counterparts, showing that none of the four modalities is fully
redundant. Among the tested removals, excluding the market modality
causes the largest reduction on both assets. For \method{}, AAPL ARR
decreases from 67.79\% to 58.42\%, while BTCUSD ARR decreases from
84.64\% to 72.18\%. Removing news or sentiment produces smaller
reductions in these two datasets, although performance still remains
below the clean setting. More importantly, the model with RPL
outperforms the variant without RPL in both ARR and SR under every
single-modality removal condition. For example, after removing the
market modality, RPL improves ARR from 51.10\% to 58.42\% on AAPL
and from 63.45\% to 72.18\% on BTCUSD. These comparisons indicate
that RPL reduces, but does not eliminate, the loss caused by a
completely unavailable modality.

The stochastic results show a similar pattern. ARR and SR decline
progressively as the dropout probability increases for both model
variants, reflecting the increasing frequency of missing-modality
instances. However, the performance gap between \method{} and the
variant without RPL generally becomes larger at higher dropout
rates. At $p=1.00$, \method{} achieves an ARR of 54.12\% and an SR
of 1.08 on AAPL, compared with 42.35\% and 0.76 without RPL. On
BTCUSD, the corresponding results are 68.10\% and 1.98 with RPL,
versus 51.40\% and 1.21 without RPL. Therefore, the results support
the narrower conclusion that RPL improves tolerance to the examined
single-missing-modality conditions, particularly when modality
unavailability becomes frequent. They do not establish robustness
to multiple simultaneous missing modalities or arbitrary forms of
input corruption.




\begin{figure*}[t]
    \centering
    \includegraphics[width=0.98\textwidth]{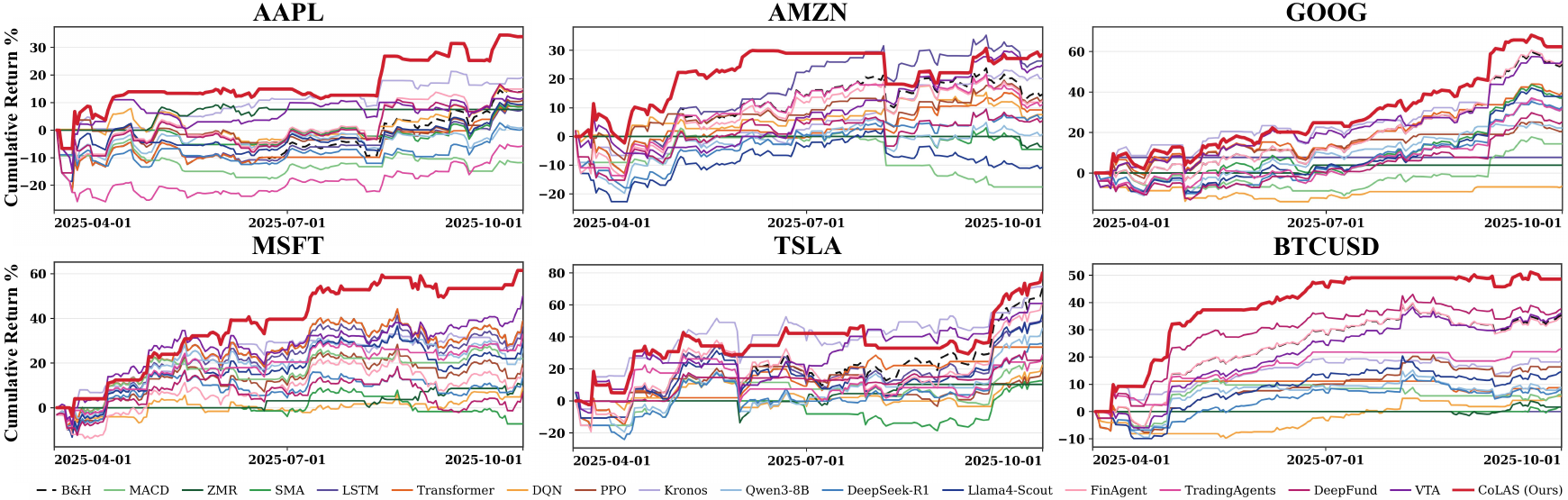}
    \caption{Performance comparison on cumulative return (CR\%$\uparrow$) between \method{} and other baselines over all assets.}
    \Description{Main experiment}
    \label{fig:cr_curve_all}
\end{figure*}

\subsection{Cumulative Return Analyses}
\label{app:cr_curve}
Figure~\ref{fig:cr_curve_all} compares the cumulative-return
trajectories of \method{} and the baselines across the six primary
assets. \method{} achieves the highest terminal cumulative return on
all six assets, although it does not remain above every baseline at
all points during the evaluation period.

The return trajectories exhibit different patterns across assets.
On AAPL and BTCUSD, a substantial portion of the final return is
accumulated through several pronounced upward movements, followed by
relatively stable intervals. On GOOG and MSFT, the gains are
accumulated through multiple stages, with temporary pullbacks before
the end of the evaluation period. AMZN and TSLA exhibit more
fluctuating trajectories, indicating that the advantage of
\method{} is not associated with a uniformly smooth return path.
Nevertheless, \method{} finishes the evaluation period ahead of the
baselines on both assets.

The step-like segments observed on some assets are consistent with
returns being concentrated in a limited number of periods. However,
cumulative-return curves alone do not identify trading frequency,
turnover, or position exposure. Therefore, these curves should not
be interpreted as direct evidence that \method{} trades only when
the modalities corroborate. Such a mechanism would require a
separate analysis of trading actions or portfolio positions.

%% file: tables/notation.tex

\begin{table}[t]
\centering
\scriptsize
\caption{Summary of the principal notation used in CoLAS.}
\label{tab:notation}
\renewcommand{\arraystretch}{1.04}
\setlength{\tabcolsep}{3pt}
\begin{tabularx}{\columnwidth}{
  @{}
  >{\centering\arraybackslash}p{0.25\columnwidth}
  X
  @{}
}
\toprule
\textbf{Symbol} & \textbf{Definition} \\
\midrule
\multicolumn{2}{@{}l}{\textit{Data and indices}} \\
$\mathcal{X}$ &
$\{(\{\mathbf{X}_i^m\}_{m\in\mathcal{M}},y_i)\}_{i=1}^{N}$,
the training dataset \\
$\mathcal{M}$; $k$ &
Modality set; $k=|\mathcal{M}|$ \\
$N$; $B$ &
Numbers of instances and mini-batch instances \\
$i,n$ &
Instance indices; $n$ is the second batch index \\
$m$; $m^\star$ &
Modality index; randomly perturbed modality \\
$j$ &
Singular-component index, $j\in\{1,\ldots,k\}$ \\
$T_{\mathrm{win}}$; $\tau$ &
Look-back-window length; within-window time index \\
$\mathbf{X}_i^m$; $y_i$ &
Input sequence
$\{\mathbf{x}_{i,\tau}^m\}_{\tau=1}^{T_{\mathrm{win}}}$;
binary direction label \\
\midrule
\multicolumn{2}{@{}l}{\textit{Multimodal embedding}} \\
$f_m(\cdot;\boldsymbol{\theta}_m)$ &
Encoder for modality $m$ with parameters
$\boldsymbol{\theta}_m$ \\
$\mathbf{h}_i^m$ &
Modality-specific embedding
$f_m(\mathbf{X}_i^m;\boldsymbol{\theta}_m)$ \\
$\psi_m$ &
Modality-specific alignment map \\
$\mathbf{v}_i^m$;
$\bar{\mathbf{v}}_i^m$ &
Unnormalized and $\ell_2$-normalized aligned
representations in $\mathbb{R}^{d_c}$ \\
$d_c$ &
Dimension of the common alignment space \\
$\mathbf{V}_i$ &
$[\bar{\mathbf{v}}_i^{m_1},\ldots,
\bar{\mathbf{v}}_i^{m_k}]
\in\mathbb{R}^{d_c\times k}$ \\
$\mathbf{G}_i$ &
Gram matrix
$\mathbf{V}_i^\top\mathbf{V}_i
\in\mathbb{R}^{k\times k}$ \\
\midrule
\multicolumn{2}{@{}l}{\textit{Corroboration mining}} \\
$\mathbf{P}_i\boldsymbol{\Lambda}_i
\mathbf{R}_i^\top$ &
Compact SVD of $\mathbf{V}_i$, with
$\boldsymbol{\Lambda}_i=
\operatorname{diag}(\lambda_{i,1},\ldots,\lambda_{i,k})$ \\
$\lambda_{i,j}$ &
$j$-th singular value; 
$\lambda_{i,1}\ge\cdots\ge\lambda_{i,k}\ge0$ \\
$\mathbf{p}_{i,j}$;
$\mathbf{r}_{i,j}$ &
$j$-th left and right singular vectors \\
$\mathbf{V}_i^{(1)}$ &
Leading rank-1 component
$\lambda_{i,1}\mathbf{p}_{i,1}\mathbf{r}_{i,1}^{\top}$ \\
$a_{i,m}$ &
Signed modality contribution
$\mathbf{p}_{i,1}^{\top}\bar{\mathbf{v}}_i^m$ \\
$s_i$ &
Net signed support
$\sum_{m=1}^{k}a_{i,m}$ \\
$\mathbf{1}_k$ &
All-ones vector in $\mathbb{R}^{k}$ \\
$\mathbf{z}_i$ &
Corroborated signal
$k^{-1}\mathbf{V}_i^{(1)}\mathbf{1}_k
=k^{-1}s_i\mathbf{p}_{i,1}$ \\
$\rho_i$; $\delta_i$ &
Spectral concentration $\lambda_{i,1}^2/k$;
net signed support factor
$(\mathbf{r}_{i,1}^{\top}\mathbf{1}_k)^2/k$ \\
\midrule
\multicolumn{2}{@{}l}{\textit{Robust prediction and optimization}} \\
$\Pi$;
$\widetilde{\mathbf{X}}_i^{m^\star}$ &
Perturbation operator; perturbed modality input
$\Pi(\mathbf{X}_i^{m^\star})$ \\
$\boldsymbol{\epsilon}$;
$\sigma_\epsilon$ &
Gaussian noise and its standard deviation,
$\boldsymbol{\epsilon}\sim
\mathcal{N}(\mathbf{0},\sigma_\epsilon^2\mathbf{I})$ \\
$\widetilde{\mathbf{z}}_i$ &
Corroborated signal from the perturbed input \\
$g$;
$\ell_i,\widetilde{\ell}_i$ &
Prediction head; clean and perturbed logits \\
$\mathcal{L}_{\mathrm{SVM}}$;
$\mathcal{L}_{\mathrm{IR}}$ &
Singular-value-maximization and
instance-wise regularization losses \\
$\mathcal{L}_{\mathrm{rob}}$;
$\mathcal{L}_{\mathrm{cls}}$ &
Robust-consistency and binary-classification losses \\
$\mathcal{L}_{\mathrm{CoLAS}}$ &
Overall training objective \\
$\gamma_1,\gamma_2$;
$\beta_1,\beta_2,\beta_3$ &
Loss temperatures; loss weights \\
$\boldsymbol{\Theta}$; $\eta$ &
All trainable parameters; learning rate \\
\bottomrule
\end{tabularx}
\end{table}

%% file: tables/news_coverage.tex
\begin{table}[H]
\centering
\scriptsize
\caption{Dataset statistics detailing the chronological period and
the amount of each data source for the six primary assets.}
\label{tab:dataset_statistics}
\renewcommand{\arraystretch}{1.08}
\setlength{\tabcolsep}{3.5pt}

\resizebox{\columnwidth}{!}{
\begin{tabular}{lrrrrrr}
\toprule
\textbf{Asset}
& \textbf{AAPL}
& \textbf{AMZN}
& \textbf{GOOG}
& \textbf{MSFT}
& \textbf{TSLA}
& \textbf{BTCUSD} \\
\midrule

Data period
& \multicolumn{6}{c}{From 2023-10-01 to 2026-03-20} \\

\midrule

Market data
& \multicolumn{6}{c}{
Daily $\times$ (open, high, low, adjusted close, volume)
} \\

Technical data
& \multicolumn{6}{c}{
Daily $\times$ $d_{\mathrm{tech}}$ technical indicators
} \\

\midrule

News records
& 8,866
& 7,547
& 6,201
& 7,924
& 11,676
& 11,986 \\

Sentiment records
& 11,020
& 6,480
& 12,410
& 9,150
& 8,350
& 7,120 \\

\bottomrule
\end{tabular}
}
\end{table}

%% file: tables/news_prompt.tex
\begin{table*}[t]
\centering
\scriptsize
\caption{Prompt template used to extract the news representation.}
\label{tab:news_agent_prompt_example}
\renewcommand{\arraystretch}{1.08}
\setlength{\tabcolsep}{5pt}

\begin{tabularx}{\textwidth}{
    @{}
    >{\raggedright\arraybackslash\bfseries}p{0.14\textwidth}
    >{\raggedright\arraybackslash}X
    @{}}
\toprule
\textbf{Component} & \textbf{Content} \\
\midrule

\rowcolor{gray!12}
\multicolumn{2}{@{}l}{\textbf{Prompt configuration}} \\

Task
& Predict the next-day stock trading signal from recent news. \\

Input
& News summaries from the preceding $T$ trading days. \\

Output
& A concise analysis followed by a binary signal:
\texttt{Bullish} or \texttt{Bearish}. \\

\addlinespace[2pt]
\rowcolor{gray!12}
\multicolumn{2}{@{}l}{\textbf{Unified prompt with a filled example}} \\

Prompt
&
{\ttfamily
You are a financial news analyst. Given recent news about a stock,
assess its implications and predict the stock's next-day direction.
Base your analysis only on the supplied news. Provide a concise
rationale and return either Bullish or Bearish.

\medskip
Asset: AAPL

Reference date: 2025-11-03

News summaries from the preceding \(T\) trading days:

\medskip
Day 1: Apple trades steadily as BofA raises its iPhone 17 unit
forecast and price target, projecting stronger long-term revenue
and margins. Additional reports preview Apple's upcoming earnings.

\medskip
Day 2: Apple is scheduled to report fourth-quarter results after
Thursday's market close. Analysts expect earnings of \$1.77 per
share and quarterly revenue of \$102.17 billion. Major U.S.
technology companies report generally upbeat results.

\medskip
Day 3: Apple shares reach new all-time highs following strong
fourth-quarter results. The company guides for a stronger,
potentially record-breaking first quarter. Wall Street analysts
respond positively, citing stronger iPhone demand and a solid
holiday outlook.

\medskip
\ldots

\medskip
Output format:

Analysis: <concise news-based analysis>

Signal: Bullish/Bearish
}
\\

\addlinespace[2pt]
\rowcolor{gray!12}
\multicolumn{2}{@{}l}{\textbf{Response}} \\

Output
&
{\ttfamily
Analysis: The recent news flow is predominantly positive for
Apple. Strong fourth-quarter results, optimistic forward guidance,
and bullish analyst reactions indicate favorable near-term
momentum. Although broader market uncertainties remain, the
company-specific news supports a positive next-day outlook.

Signal: Bullish
}
\\

\bottomrule
\end{tabularx}
\end{table*}

%% file: tables/sentiment_prompt.tex
\begin{table*}[t]
\centering
\scriptsize
\caption{Prompt template used to extract the sentiment representation.}
\label{tab:sentiment_agent_prompt_example}
\renewcommand{\arraystretch}{1.08}
\setlength{\tabcolsep}{5pt}

\begin{tabularx}{\textwidth}{
    @{}
    >{\raggedright\arraybackslash\bfseries}p{0.14\textwidth}
    >{\raggedright\arraybackslash}X
    @{}}
\toprule
\textbf{Component} & \textbf{Content} \\
\midrule

\rowcolor{gray!12}
\multicolumn{2}{@{}l}{\textbf{Prompt configuration}} \\

Task
& Classify the overall sentiment expressed toward the target asset
in recent sentiment-related records. \\

Input
& Sentiment-related records from the preceding $T$ trading days. \\

Output
& A concise analysis followed by a binary sentiment label:
\texttt{Positive} or \texttt{Negative}. \\

\addlinespace[2pt]
\rowcolor{gray!12}
\multicolumn{2}{@{}l}{\textbf{Unified prompt with a filled example}} \\

Prompt
&
{\ttfamily
You are a financial sentiment analyst. Given recent sentiment-related
information about the target asset, assess the overall sentiment expressed
toward the target asset. Base your analysis only on the supplied
records. Provide a concise rationale and return either Positive or
Negative.

\medskip
Asset: AAPL

Reference date: 2025-11-03

Sentiment-related information for the preceding $T$ trading days:

\medskip
Day 1: Apple is scheduled to report its fourth-quarter results on
Thursday, with investor attention focused on iPhone 17 demand.
Analysts expect steady performance and maintain optimistic price
targets.

\medskip
Day 2: Apple trades steadily as BofA raises its iPhone 17 unit
forecast and price target, projecting stronger long-term revenue
and margins. Additional reports preview Apple's upcoming earnings.

\medskip
Day 3: Apple is scheduled to report fourth-quarter results after
Thursday's market close. Analysts expect earnings of \$1.77 per
share and quarterly revenue of \$102.17 billion, while U.S. stock
futures remain mixed.

\medskip
\ldots

\medskip
Output format:

Analysis: <concise sentiment analysis based only on the supplied records>

Sentiment: Positive/Negative
}
\\

\addlinespace[2pt]
\rowcolor{gray!12}
\multicolumn{2}{@{}l}{\textbf{Response}} \\

Output
&
{\ttfamily
Analysis: The supplied records are predominantly positive for Apple.
Analysts anticipate solid quarterly performance and sustained
iPhone demand, while BofA has raised its unit forecast and price
target. Although broader market conditions are mixed, no material
negative company-specific development is reported.

Sentiment: Positive
}
\\

\bottomrule
\end{tabularx}
\end{table*}

%% file: tables/broader_asset_pool.tex
\begin{table*}[t]
\centering
\caption{Extended performance on the five newly added assets
(KO, JNJ, UPS, UPWK, and SFM), evaluated over 2025-04-01 to
2026-03-20 with the original training and validation sets held fixed.
The first panel reports return-oriented metrics, where higher is
better for CR, ARR, and SR. The second panel reports risk-oriented
metrics, where higher is better for SoR and CalR, and lower is better
for MDD. Best results are \textbf{bolded}, and second-best results
are \underline{underlined}.}
\label{tab:broader_asset_pool}

\resizebox{\textwidth}{!}{
\begin{tabular}{l *{15}{c}}
\toprule
\multirow{2}{*}{\textbf{Method}}
& \multicolumn{3}{c}{\textbf{KO}}
& \multicolumn{3}{c}{\textbf{JNJ}}
& \multicolumn{3}{c}{\textbf{UPS}}
& \multicolumn{3}{c}{\textbf{UPWK}}
& \multicolumn{3}{c}{\textbf{SFM}} \\
\cmidrule(lr){2-4}
\cmidrule(lr){5-7}
\cmidrule(lr){8-10}
\cmidrule(lr){11-13}
\cmidrule(lr){14-16}
& \textbf{CR\%$\uparrow$}
& \textbf{ARR\%$\uparrow$}
& \textbf{SR$\uparrow$}
& \textbf{CR\%$\uparrow$}
& \textbf{ARR\%$\uparrow$}
& \textbf{SR$\uparrow$}
& \textbf{CR\%$\uparrow$}
& \textbf{ARR\%$\uparrow$}
& \textbf{SR$\uparrow$}
& \textbf{CR\%$\uparrow$}
& \textbf{ARR\%$\uparrow$}
& \textbf{SR$\uparrow$}
& \textbf{CR\%$\uparrow$}
& \textbf{ARR\%$\uparrow$}
& \textbf{SR$\uparrow$} \\
\midrule

\multicolumn{16}{c}{\textit{Rule-based Strategies}} \\
\midrule
B\&H
& 6.94 & 7.23 & 0.49
& \second{57.13} & \second{59.24} & 2.73
& -6.22 & -6.41 & -0.07
& -13.62 & -14.12 & 0.04
& -46.93 & -48.72 & -1.26 \\

SMA
& 1.72 & 1.83 & 0.21
& 33.64 & 35.02 & 2.23
& 5.24 & 5.43 & 0.35
& 28.63 & 29.84 & 0.89
& -24.93 & -25.92 & -1.46 \\

\midrule
\multicolumn{16}{c}{\textit{Single-modal Models}} \\
\midrule
Transformer
& 14.03 & 14.52 & 1.01
& 20.54 & 21.32 & 1.32
& 5.32 & 5.51 & 0.45
& 15.72 & 16.34 & 0.54
& -38.62 & -40.04 & -0.95 \\

PPO
& -15.03 & -15.62 & -1.19
& 10.02 & 10.43 & 0.92
& -6.13 & -6.32 & -0.29
& -15.34 & -15.93 & -0.17
& -15.63 & -16.12 & -0.69 \\

Kronos
& 14.82 & 15.43 & 1.21
& 25.13 & 26.02 & 1.69
& 5.34 & 5.52 & 0.35
& 42.94 & 44.52 & \second{1.13}
& -25.42 & -26.33 & -0.56 \\

\midrule
\multicolumn{16}{c}{\textit{Multimodal General LLMs}} \\
\midrule
Llama4-Scout-17B
& -6.42 & -6.63 & -0.49
& 43.72 & 45.43 & \second{2.77}
& -18.63 & -19.32 & -0.70
& -22.23 & -23.02 & -0.53
& -39.42 & -40.93 & -1.10 \\

\midrule
\multicolumn{16}{c}{\textit{Multimodal Financial LLMs}} \\
\midrule
FinAgent
& 5.12 & 5.33 & 0.40
& 52.93 & 54.82 & 2.59
& \second{12.14} & \second{12.63} & \second{0.64}
& 3.42 & 3.51 & 0.34
& -28.12 & -29.13 & -0.86 \\

TradingAgents
& -3.62 & -3.81 & -0.52
& 22.73 & 23.62 & 2.76
& -17.52 & -18.13 & -0.81
& -16.53 & -17.12 & -0.24
& -54.92 & -56.93 & -2.10 \\

DeepFund
& -10.02 & -10.42 & -0.76
& 17.03 & 17.62 & 1.19
& -5.12 & -5.31 & -0.09
& -17.82 & -18.53 & -0.14
& \second{-5.53} & \second{-5.72} & \second{-0.07} \\

VTA
& \second{22.74} & \second{23.52} & \second{1.54}
& 35.02 & 36.23 & 2.40
& 8.03 & 8.32 & 0.47
& \second{43.02} & \second{44.63} & 1.05
& -38.32 & -39.83 & -0.98 \\

\midrule
\rowcolor{oursblue}
\method{}~(Ours)
& \best{34.63} & \best{35.82} & \best{2.43}
& \best{60.34} & \best{62.53} & \best{2.87}
& \best{25.32} & \best{26.24} & \best{1.71}
& \best{54.83} & \best{56.82} & \best{1.46}
& \best{15.62} & \best{16.13} & \best{0.81} \\

\bottomrule
\end{tabular}
}

\vspace{0.1em}

\resizebox{\textwidth}{!}{
\begin{tabular}{l *{15}{c}}
\toprule
\multirow{2}{*}{\textbf{Method}}
& \multicolumn{3}{c}{\textbf{KO}}
& \multicolumn{3}{c}{\textbf{JNJ}}
& \multicolumn{3}{c}{\textbf{UPS}}
& \multicolumn{3}{c}{\textbf{UPWK}}
& \multicolumn{3}{c}{\textbf{SFM}} \\
\cmidrule(lr){2-4}
\cmidrule(lr){5-7}
\cmidrule(lr){8-10}
\cmidrule(lr){11-13}
\cmidrule(lr){14-16}
& \textbf{SoR$\uparrow$}
& \textbf{CalR$\uparrow$}
& \textbf{MDD\%$\downarrow$}
& \textbf{SoR$\uparrow$}
& \textbf{CalR$\uparrow$}
& \textbf{MDD\%$\downarrow$}
& \textbf{SoR$\uparrow$}
& \textbf{CalR$\uparrow$}
& \textbf{MDD\%$\downarrow$}
& \textbf{SoR$\uparrow$}
& \textbf{CalR$\uparrow$}
& \textbf{MDD\%$\downarrow$}
& \textbf{SoR$\uparrow$}
& \textbf{CalR$\uparrow$}
& \textbf{MDD\%$\downarrow$} \\
\midrule

\multicolumn{16}{c}{\textit{Rule-based Strategies}} \\
\midrule
B\&H
& 0.83 & 0.74 & 9.72
& \second{4.33} & 7.12 & 8.32
& -0.09 & -0.29 & 22.23
& 0.06 & -0.29 & 48.62
& -1.38 & -0.77 & 63.14 \\

SMA
& 0.26 & 0.22 & 8.43
& 3.56 & 6.72 & 7.82
& 0.35 & 0.23 & 23.82
& 1.11 & 1.13 & 26.43
& -1.16 & -0.74 & \second{35.03} \\

\midrule
\multicolumn{16}{c}{\textit{Single-modal Models}} \\
\midrule
Transformer
& 1.76 & 1.23 & 11.83
& 1.76 & 3.09 & 6.92
& 0.51 & 0.61 & \second{9.04}
& 0.60 & 0.41 & 39.73
& -0.99 & -0.66 & 60.92 \\

PPO
& -1.48 & -0.79 & 19.62
& 1.02 & 1.66 & \second{6.32}
& -0.38 & -0.49 & 12.93
& -0.20 & -0.38 & 41.82
& -0.79 & -0.44 & 36.32 \\

Kronos
& 1.91 & 2.12 & 7.72
& 2.22 & 2.65 & 9.83
& 0.35 & 0.29 & 18.92
& \second{1.62} & \second{2.67} & \second{16.73}
& -0.54 & -0.54 & 48.72 \\

\midrule
\multicolumn{16}{c}{\textit{Multimodal General LLMs}} \\
\midrule
Llama4-Scout-17B
& -0.70 & -0.43 & 15.42
& 4.27 & \second{7.14} & 7.12
& -0.76 & -0.70 & 27.63
& -0.69 & -0.64 & 36.04
& -1.10 & -0.73 & 56.02 \\

\midrule
\multicolumn{16}{c}{\textit{Multimodal Financial LLMs}} \\
\midrule
FinAgent
& 0.64 & 0.49 & 10.92
& 4.04 & 6.58 & 8.32
& \second{0.93} & \second{0.87} & 14.43
& 0.53 & 0.08 & 45.42
& -1.31 & -0.57 & 50.73 \\

TradingAgents
& -0.49 & -0.55 & 7.92
& 3.46 & 6.96 & \best{3.42}
& -0.71 & -0.83 & 21.93
& -0.23 & -0.43 & 39.42
& -1.61 & -0.92 & 61.83 \\

DeepFund
& -1.05 & -0.73 & 14.23
& 1.56 & 1.56 & 11.32
& -0.10 & -0.22 & 24.23
& -0.19 & -0.40 & 46.32
& \second{-0.10} & \second{-0.14} & 41.83 \\

VTA
& \second{2.74} & \second{3.69} & \best{6.42}
& 3.93 & 6.56 & 6.82
& 0.50 & 0.51 & 16.42
& 1.41 & 1.34 & 33.42
& -0.97 & -0.69 & 57.92 \\

\midrule
\rowcolor{oursblue}
\method{}~(Ours)
& \best{3.34} & \best{4.77} & \second{7.52}
& \best{4.52} & \best{7.52} & \second{6.32}
& \best{4.13} & \best{6.71} & \best{3.92}
& \best{2.87} & \best{3.57} & \best{15.93}
& \best{1.26} & \best{0.94} & \best{17.23} \\

\bottomrule
\end{tabular}
}
\end{table*}

%% file: tables/extended_test_horizons.tex
\begin{table*}[t]
\centering
\caption{Extended horizon performance on four representative assets
from the main experiments: AAPL, GOOG, TSLA, and BTCUSD. The
evaluation period is 2025-04-01 to 2026-03-20, while the original
training and validation datasets remain unchanged. Best results are
\textbf{bolded}, and second-best results are
\underline{underlined}. Rankings are determined using the unrounded
values.}
\label{tab:extended_test_horizon}

\resizebox{\textwidth}{!}{
\begin{tabular}{l *{12}{N}}
\toprule
\multirow{2}{*}{\textbf{Method}}
& \multicolumn{3}{c}{\textbf{AAPL}}
& \multicolumn{3}{c}{\textbf{GOOG}}
& \multicolumn{3}{c}{\textbf{TSLA}}
& \multicolumn{3}{c}{\textbf{BTCUSD}} \\
\cmidrule(lr){2-4}
\cmidrule(lr){5-7}
\cmidrule(lr){8-10}
\cmidrule(lr){11-13}
& \textbf{CR\%$\uparrow$}
& \textbf{ARR\%$\uparrow$}
& \textbf{SR$\uparrow$}
& \textbf{CR\%$\uparrow$}
& \textbf{ARR\%$\uparrow$}
& \textbf{SR$\uparrow$}
& \textbf{CR\%$\uparrow$}
& \textbf{ARR\%$\uparrow$}
& \textbf{SR$\uparrow$}
& \textbf{CR\%$\uparrow$}
& \textbf{ARR\%$\uparrow$}
& \textbf{SR$\uparrow$} \\
\midrule

\multicolumn{13}{c}{\textit{Rule-based Strategies}} \\
\midrule
B\&H
& 11.06 & 11.46 & 0.49
& 86.95 & 90.17 & 2.32
& 36.87 & 38.23 & 0.85
& -16.96 & -17.58 & -0.22 \\

SMA
& -16.48 & -17.09 & -0.90
& 57.60 & 59.74 & 2.10
& -6.82 & -7.07 & 0.01
& 2.27 & 2.35 & 0.21 \\

\midrule
\multicolumn{13}{c}{\textit{Single-modal Models}} \\
\midrule
Transformer
& \second{22.87} & \second{23.71} & \second{0.88}
& 47.94 & 49.72 & 1.89
& \second{59.06} & \second{61.25} & \second{1.37}
& -16.96 & -17.58 & -0.22 \\

PPO
& 0.20 & 0.20 & 0.14
& 35.84 & 37.17 & 1.76
& 24.67 & 29.29 & 1.24
& -28.50 & -29.47 & -0.64 \\

Kronos
& -9.09 & -9.43 & -0.23
& 38.23 & 39.65 & 1.54
& 24.39 & 25.30 & 0.72
& -18.27 & -18.89 & -0.42 \\

\midrule
\multicolumn{13}{c}{\textit{Multimodal General LLMs}} \\
\midrule
Llama4-Scout-17B
& 0.31 & 0.32 & 0.15
& 68.74 & 71.27 & 2.24
& 10.50 & 10.84 & 0.45
& -6.69 & -6.93 & -0.05 \\

\midrule
\multicolumn{13}{c}{\textit{Multimodal Financial LLMs}} \\
\midrule
FinAgent
& 17.61 & 18.19 & 0.80
& \second{89.51} & \second{92.45} & \second{2.37}
& 28.33 & 29.26 & 0.73
& -19.31 & -19.91 & -0.41 \\

TradingAgents
& -13.55 & -14.00 & -0.50
& 52.02 & 53.73 & 2.00
& 46.04 & 47.55 & 1.29
& \second{30.90} & \second{31.86} & \second{1.14} \\

DeepFund
& 14.32 & 14.79 & 0.62
& 47.47 & 49.03 & 1.76
& 32.26 & 33.31 & 0.83
& -39.05 & -40.26 & -1.24 \\

VTA
& -7.28 & -7.71 & -0.15
& 45.99 & 48.70 & 1.85
& 31.24 & 33.08 & 0.82
& -12.28 & -12.81 & -0.16 \\

\midrule
\rowcolor{oursblue}
\method{}~(Ours)
& \best{33.24} & \best{34.47} & \best{1.30}
& \best{96.09} & \best{99.65} & \best{2.69}
& \best{69.05} & \best{71.61} & \best{1.60}
& \best{46.71} & \best{48.30} & \best{1.62} \\

\bottomrule
\end{tabular}
}

\vspace{0.1em}

\resizebox{\textwidth}{!}{
\begin{tabular}{l *{12}{N}}
\toprule
\multirow{2}{*}{\textbf{Method}}
& \multicolumn{3}{c}{\textbf{AAPL}}
& \multicolumn{3}{c}{\textbf{GOOG}}
& \multicolumn{3}{c}{\textbf{TSLA}}
& \multicolumn{3}{c}{\textbf{BTCUSD}} \\
\cmidrule(lr){2-4}
\cmidrule(lr){5-7}
\cmidrule(lr){8-10}
\cmidrule(lr){11-13}
& \textbf{SoR$\uparrow$}
& \textbf{CalR$\uparrow$}
& \textbf{MDD\%$\downarrow$}
& \textbf{SoR$\uparrow$}
& \textbf{CalR$\uparrow$}
& \textbf{MDD\%$\downarrow$}
& \textbf{SoR$\uparrow$}
& \textbf{CalR$\uparrow$}
& \textbf{MDD\%$\downarrow$}
& \textbf{SoR$\uparrow$}
& \textbf{CalR$\uparrow$}
& \textbf{MDD\%$\downarrow$} \\
\midrule

\multicolumn{13}{c}{\textit{Rule-based Strategies}} \\
\midrule
B\&H
& 0.76 & 0.50 & 22.77
& 4.04 & 6.74 & 13.38
& 1.31 & 1.55 & 24.66
& -0.30 & -0.36 & 49.36 \\

SMA
& -1.11 & -0.91 & 18.85
& 3.59 & 6.68 & 8.94
& 0.01 & -0.19 & 36.36
& 0.31 & 0.08 & 29.40 \\

\midrule
\multicolumn{13}{c}{\textit{Single-modal Models}} \\
\midrule
Transformer
& \second{1.41} & 1.04 & 22.77
& 3.15 & 6.08 & 8.18
& \second{2.48} & 2.87 & 21.33
& -0.30 & -0.36 & 49.36 \\

PPO
& 0.21 & 0.01 & 22.77
& 3.34 & 5.22 & \best{7.12}
& 1.64 & \second{3.92} & \second{11.84}
& -0.89 & -0.41 & 49.09 \\

Kronos
& -0.35 & -0.41 & 22.80
& 2.64 & 3.20 & 12.40
& 1.13 & 1.19 & 21.33
& -0.57 & -0.50 & 38.10 \\

\midrule
\multicolumn{13}{c}{\textit{Multimodal General LLMs}} \\
\midrule
Llama4-Scout-17B
& 0.22 & 0.02 & 18.84
& 3.78 & 6.95 & 10.22
& 0.64 & 0.35 & 30.60
& -0.07 & -0.21 & 33.07 \\

\midrule
\multicolumn{13}{c}{\textit{Multimodal Financial LLMs}} \\
\midrule
FinAgent
& 1.21 & \second{1.32} & \second{13.75}
& \second{4.13} & \second{7.29} & 12.68
& 1.13 & 1.18 & 24.85
& -0.57 & -0.48 & 41.72 \\

TradingAgents
& -0.64 & -0.53 & 26.17
& 3.27 & 5.40 & 9.95
& 1.85 & 2.11 & 22.52
& \second{1.52} & \second{1.33} & \second{24.00} \\

DeepFund
& 0.98 & 0.94 & 15.78
& 2.99 & 3.69 & 13.29
& 1.39 & 0.92 & 36.02
& -1.59 & -0.80 & 50.07 \\

VTA
& -0.23 & -0.34 & 22.77
& 3.73 & 5.53 & 8.81
& 1.30 & 1.16 & 28.64
& -0.22 & -0.29 & 43.73 \\

\rowcolor{oursblue}
\midrule
\method{}~(Ours)
& \best{2.52} & \best{3.91} & \best{8.81}
& \best{5.02} & \best{8.25} & \second{8.08}
& \best{3.23} & \best{6.23} & \best{11.50}
& \best{2.72} & \best{2.84} & \best{17.01} \\

\bottomrule
\end{tabular}
}

\end{table*}

%% file: tables/app_ablation_study.tex
\begin{table}[t]
  \centering
  \caption{Different multimodal representation
  constructions on AAPL and BTCUSD. ARR denotes the annual rate
  of return, and SR denotes the Sharpe ratio.}
  \label{tab:ablation_study_2}
  \setlength{\tabcolsep}{4.5pt}
  \renewcommand{\arraystretch}{1.10}
  \small
  \begin{tabular}{lcccc}
    \toprule
    \multirow{2}{*}{\textbf{Method}}
      & \multicolumn{2}{c}{\textbf{AAPL}}
      & \multicolumn{2}{c}{\textbf{BTCUSD}} \\
    \cmidrule(lr){2-3}
    \cmidrule(lr){4-5}
      & \textbf{ARR\%$\uparrow$} & \textbf{SR$\uparrow$}
      & \textbf{ARR\%$\uparrow$} & \textbf{SR$\uparrow$} \\
    \midrule
    Mean Pooling
      & 18.95 & 0.54 & 21.10 & 1.09 \\
    Concatenation
      & 16.41 & 0.90 & 23.22 & 0.91 \\
    \method{}$_{FSC}$
      & 37.46 & 0.77 & \second{60.24} & \second{1.63} \\
    \method{}$_{unsigned}$
      & \second{41.23} & \second{1.20} & 58.42 & 1.59 \\
    \midrule
    \textbf{\method{}}
      & \textbf{67.79} & \textbf{1.47}
      & \textbf{84.64} & \textbf{2.65} \\
    \bottomrule
  \end{tabular}
\end{table}